\documentclass[12pt,twoside]{article} 
\usepackage[T1]{fontenc} % font

\usepackage{amsthm, amsmath, amssymb, mathrsfs, upgreek, enumerate, mathtools, comment, natbib, graphicx, mathabx, booktabs, threeparttable, relsize, exscale, epstopdf, scrextend, multirow, xr-hyper, pdfpages, bm, caption, subcaption, tikz}
\usepackage[headheight=110pt,margin=1.25in]{geometry}
\usepackage[section]{placeins} % for \FloatBarrier command

\usepackage{fancyhdr, setspace} % headers and footers
\usepackage[pagebackref=false]{hyperref} % link colors
\hypersetup{
    colorlinks=true,
    citecolor=blue,
    filecolor=black,
    linkcolor=black,
    urlcolor=blue,
    bookmarksopen=true,
    pdfstartview=FitH
}

\newtheorem{theorem}{Theorem}
\newtheorem{assump}{Assumption}
\newtheorem{seclemma}{Lemma}[section]

\newtheorem{secalgo}{Algorithm}[section]

\theoremstyle{definition}

\newtheorem{example}{Example}

\newtheorem{remark}{Remark}

\def\Snospace~{\S{}}
\makeatletter
\def\thm@space@setup{
  \thm@preskip=15pt \thm@postskip=15pt % controls spacing before and
}
\makeatother

\def\indep{\perp\!\!\!\perp}

\newcommand{\argmin}{\operatornamewithlimits{argmin}}

\newcommand{\cov}{\text{Cov}}
\newcommand{\var}{\text{Var}}
\newcommand{\E}{{\bf E}}
\newcommand{\R}{\mathbb{R}}
\newcommand{\F}{\mathcal{F}}
\newcommand{\N}{\mathcal{N}}
\newcommand{\C}{\mathcal{C}}
\newcommand{\X}{\mathcal{X}}

\newcommand{\plimarrow}{\stackrel{p}\longrightarrow}
\newcommand{\dlimarrow}{\stackrel{d}\longrightarrow}
\newcommand{\ind}{\bm{1}}
\providecommand{\abs}[1]{\lvert#1\rvert}

\let\emptyset\varnothing
\providecommand{\abs}[1]{\lvert#1\rvert}

\renewcommand{\qed}{\hfill \mbox{\raggedright \rule{0.08in}{0.08in}}} % black QED box
\title{\bf\sc Cluster-Randomized Trials with Cross-Cluster Interference \thanks{This paper was previously circulated under the title, ``Design of Cluster-Randomized Trials with Cross-Cluster Interference.'' I thank Dennis Egger for generously providing supplemental distance data used in the empirical application. I also thank Forrest Crawford and seminar participants at Texas A\&M and the UCSC brown bag for helpful comments.}}

\author{Michael P.\ Leung\thanks{Department of Economics, University of California, Santa Cruz. E-mail: leungm@ucsc.edu.}}

\begin{document}
\maketitle
\onehalfspacing
 
\begin{abstract}

  {\sc Abstract.} The literature on cluster-randomized trials typically allows for interference within but not across clusters. This may be implausible when units are irregularly distributed across space without well-separated communities, as clusters in such cases may not align with significant geographic, social, or economic divisions. This paper develops methods for reducing bias due to cross-cluster interference. We first propose an estimation strategy that excludes units not surrounded by clusters assigned to the same treatment arm. We show that this substantially reduces bias relative to conventional difference-in-means estimators without significant cost to variance. Second, we formally establish a bias-variance trade-off in the choice of clusters: constructing fewer, larger clusters reduces bias due to interference but increases variance. We provide a rule for choosing the number of clusters to balance the asymptotic orders of the bias and variance of our estimator. Finally, we consider unsupervised learning for cluster construction and provide theoretical guarantees for $k$-medoids. 

  \bigskip


\end{abstract}

\addcontentsline{toc}{part}{Main Paper}
\newpage

%----------------------------------------------------------------------
\section{Introduction}\label{sintro}
%----------------------------------------------------------------------

The literature on cluster-randomized trials (CRTs) predominantly assumes partial interference, which allows for interference within but not across clusters. However, researchers often conduct CRTs in what \cite{hayes2017cluster} refer to as ``arbitrary geographical zones,'' where clusters are not given by nature since ``the population is widely scattered and not divided up into clearly defined and well separated communities.'' A well-known concern in these settings is {\em cross-cluster interference}, where units within a cluster may respond to treatments assigned to units in adjacent clusters. Some references refer to this phenomenon as {\em contamination} \citep{hudgens2008toward,staples2015incorporating}, while others define contamination to mean that control units procure treatment from treated neighbors, which is suggestive of interference \citep{hayes2017cluster}. In either case, a simple comparison of treatment and control clusters may be biased.

\begin{example}[Infectious disease trials]
  CRTs are widely used in the infection control and hospital epidemiology literature \citep{o2019design}. An example is the SolarMal trial \citep{homan2016effect}, designed to evaluate the impact of mosquito trapping on malaria transmission. The CRT was conducted on an island of Kenya, partitioned into contiguous clusters using an unsupervised learning algorithm. As noted by \cite{jarvis2017spatial} in a review of infectious disease trials, partial interference ``can be violated due to movement of people or diseases across borders, such as mosquitoes flying between control and intervention households.''
\end{example}

\begin{example}[Large-scale social experiments]\label{eegger}
  CRTs conducted on large populations have become increasingly common in academia \citep{muralidharan2017experimentation} and industry \citep{karrer2021network}. For example, \cite{egger2022general} study the general equilibrium effects of unconditional cash transfers using a CRT conducted in rural Kenya. The trial clustered villages into administrative units called ``sublocations.'' The authors observe that ``villages are relatively close to each other [and] sublocation boundaries are not `hard' in any sense nor reflective of salient ethnic or social divides\ldots there is extensive economic interaction in nearby markets regardless of sublocation.''
\end{example}

This paper proposes methods for reducing bias due to cross-cluster interference in the analysis and design of cluster-randomized trials. Instead of partial interference, we consider a spatial interference model proposed by \cite{leung2022rate} which posits that interference decays with geographic distance. This captures the central concern expressed in the previous examples, namely potential interference between geographically proximate units.

At the analysis stage, a common bias-reduction strategy is the ``fried-egg design'' \citep[][Ch.\ 4.4.3]{hayes2017cluster}. This is not an experimental design but rather entails excluding from estimation units in the trial located near cluster boundaries, the ``whites'' of the ``fried eggs'' that are the clusters. An open question is whether excluding observations in such a fashion is worth the loss of efficiency. We show that conventional difference-in-means estimators can have large biases due to interference near cluster boundaries. We then propose to improve the efficiency of the fried-egg design by excluding not all units near cluster boundaries but rather only the subset not surrounded by clusters assigned to the same treatment arm. We prove that this can substantially reduce the asymptotic order of the bias relative to difference in means with no increase in the asymptotic order of the variance. 

We then turn to optimal design of clusters. Partitioning the region into a small number of large clusters reduces bias since fewer units are located near cluster boundaries where they are most prone to cross-cluster interference. However, this comes at the cost of higher variance because the sample size in a CRT is the number of clusters. Unlike the conventional partial interference framework, our spatial interference model induces a bias-variance trade-off, which enables the study of cluster design for optimizing the trade-off. 

We prove that the number of clusters $k$ that balances the asymptotic bias and variance of our estimators depends on a parameter $\gamma$ that measures the speed at which interference decays with distance. This formally characterizes how domain knowledge of interference informs optimal cluster construction. Such knowledge is implicitly used in practice when researchers define fried-egg boundaries or construct clusters with ``buffer zones'' to minimize interaction between units \citep{hayes2017cluster}.

In practice, cluster construction is {\em ad hoc}. \cite{homan2016effect} utilize a ``traveling salesman algorithm,'' essentially an unsupervised learning method. More commonly, researchers use administrative divisions \citep[e.g.][]{egger2022general} or manually partition the study region \citep[e.g.][]{moulton2001design}. \cite{binka1998impact} study a malaria trial conducted in Northern Ghana, noting that, ``Where possible, small paths or road were used to delineate the clusters. However, in most cases, the cluster boundaries did not correspond to natural barriers.'' We contribute to the literature by providing formal justification for constructing clusters using $k$-medoids, a well-known unsupervised learning algorithm. 

We study CRTs under a design-based framework, as in \cite{hudgens2008toward}, \cite{imai2009essential}, and \cite{schochet2022design}, among others. Unlike these papers, we do not assume partial interference and therefore require a different formulation of standard estimands. 

Our work is most closely related to \cite{leung2022rate}. He studies designs targeting the global average treatment effect (GATE) in which clusters are squares in $\R^2$ with identical areas. We consider a more general set of estimands and richer designs that may utilize unsupervised learning to construct clusters. Also, to accommodate denser spatial settings, we consider ``infill-increasing'' asymptotics in which the number of units is of larger asymptotic order than the volume of the study region. Finally, we propose new standard errors that are asymptotically conservative without restrictions on the superpopulation. 

Our analysis of $k$-medoids builds on \cite{cao2024inference}. We extend a key result of theirs to the case where $k$ diverges and characterize other $k$-medoid properties in this regime which may be of independent interest. We generalize their Ahlfours-regularity condition on the metric space to allow for infill-increasing asymptotics, though unlike them, we also require an additional boundary condition.

\cite{faridani2023rate} substantially generalize the theoretical results in \cite{leung2022rate} while retaining his focus on the GATE. Their results hold for general spaces defined using topological conditions that differ from Ahlfours regularity. For this reason, their proofs differ substantially from ours. They study designs that essentially cover the study region with non-intersecting balls of approximate radius $g_n$, where the radius is chosen to grow at an optimal rate. Our design specifies an optimal choice for the number of clusters and uses unsupervised learning to construct clusters.

The paper is organized as follows. The next section defines the model and estimands. In \autoref{sests}, we discuss the disadvantages of existing estimators and propose an alternative. We study the theoretical properties of the estimators in \autoref{sasymp} and propose an optimal design in \autoref{sopt}. We present simulation results in \autoref{smc} and an empirical application in \autoref{sapp} using data from the \cite{egger2022general} trial. Finally, \autoref{sconclude} concludes and summarizes the practical outputs of our analysis.

We will use the following asymptotic order notation. For two sequences of random variables $\{X_n\}_{n\in\mathbb{N}}$ and $\{Y_n\}_{n\in\mathbb{N}}$, we write $X_n \precsim Y_n$ if $\abs{X_n/Y_n} = O_p(1)$, $X_n \prec Y_n$ if $\abs{X_n/Y_n} = o_p(1)$, $X_n \succsim Y_n$ if $\abs{Y_n/X_n} = O_p(1)$, $X_n \succ Y_n$ if $\abs{Y_n/X_n} = o_p(1)$, and $X_n \sim Y_n$ if both $X_n \precsim Y_n$ and $X_n \succsim Y_n$. For two sequences of constants, we use the same notation for analogous notions of asymptotic boundedness and domination.

%----------------------------------------------------------------------
\section{Setup}\label{smodel}
%----------------------------------------------------------------------

Let $(\X,\rho)$ be a metric space where $\X$ is the set of spatial locations and $\rho$ the metric. We observe a set of $n$ units $\N_n \subseteq \X$, so $\rho(i,j)$ is the spatial distance between units $i,j\in\N_n$. Let $D_i$ denote unit $i$'s binary treatment assignment and $\bm{D} = (D_i)_{i\in\N_n}$ the vector of observed assignments, which is the only random quantity in our analysis. Let $\{Y_i(\cdot)\}_{i\in\N_n}$ be a set of functions with domain $\{0,1\}^n$ and range $\R$. For $\bm{d} = (d_i)_{i\in\mathcal{N}_n} \in \{0,1\}^n$, $Y_i(\bm{d})$ denotes the potential outcome of unit $i$ under the counterfactual that all units are assigned treatments according to $\bm{d}$. Unit $i$'s observed outcome is $Y_i = Y_i(\bm{D})$. We maintain the following standard assumption, that potential outcomes are uniformly asymptotically bounded. All asymptotic statements are with respect to a sequence indexed by $n\in\mathbb{N}$ unless otherwise indicated.

\begin{assump}[Bounded Outcomes]\label{aboundY}
  $\max_{i\in\N_n} \max_{\bm{d}\in\{0,1\}^n} \abs{Y_i(\bm{d})} \precsim 1$.
\end{assump}

%---------------------------------------
\subsection{Spatial Interference}\label{sinter}
%---------------------------------------

The primary metric space of interest is $\R^2$, but our results apply to a more general ``Ahlfors-regular'' space, augmented with a boundary condition. Define the {\em $r$-neighborhood} of unit $i$
\begin{equation*}
  \N(i,r) = \{j\in\N_n\colon \rho(i,j) \leq r\}. 
\end{equation*}

\begin{assump}[Metric Space]\label{aspace}
  There exist constants $C,d>0$ and a positive sequence $\{\xi_n\}_{n\in\mathbb{N}}$ such that $1 \precsim \xi_n \prec n$ and the following hold. (a) For any $n\in\mathbb{N}$, $i\in\N_n$, and $r>0$, $\min\{C^{-1} \xi_n r^d, n\} \leq \abs{\N(i,r)} \leq \max\{C \xi_n r^d, 1\}$. (b) $d\geq 1$, and $\max_{i\in\N_n} \abs{\N(i,r+1) \backslash \N(i,r)} \leq C\max\{\xi_n r^{d-1},1\}$ for any $r>0$. 
\end{assump}

\noindent The constant $d$ represents the spatial dimension and $\xi_n$ a measure of density (units per unit volume). The dependence on $\xi_n$ is new to this paper relative to prior work on spatial interference and accommodates applications with denser regions. To understand the assumption, consider the standard {\em increasing-domain} case in which $\xi_n = 1$ for all $n$. Part (a) says that the number of units in an $r$-neighborhood is the same order $r^d$ as the neighborhood's volume. This defines a $(C,d)$-finite Ahlfors-regular space, the same metric space studied by \cite{cao2024inference}. The boundary condition in part (b) is new relative to their setup, but both (b) and the upper bound in (a) are satisfied if $\X=\R^d$ under the usual increasing-domain assumption that units are minimally separated in space \citep[][Lemma A.1]{jenish2009central}.

\autoref{aspace} also accommodates the {\em infill-increasing} case in which $\xi_n$ diverges with $n$. Here the number of units in any neighborhood is of larger order than the neighborhood volume by a factor of $\xi_n$, corresponding to a denser region. Because we require $\xi_n \prec n$, this is a hybrid of infill and increasing-domain asymptotics in which the volume of the study region $\N_n$ grows but at a slower rate $n/\xi_n$ than the population size $n$ \citep[as in][for example]{lahiri2006resampling}. 

Under \autoref{aspace}(a), the number of units in any $r$-ball is proportional to $\xi_n r^d$, which allows for some variation in density across the study region. The assumption is violated if two subregions of similar volume exist but the number of units in one is a large multiple of the other. This could be the case in practice if the study region encompasses urban and rural areas. In \autoref{shetdens}, we discuss how our proposed methodology can be modified to accommodate large variations in density.

\begin{assump}[Interference]\label{aani}
  There exist $c>0$ and $\gamma > d$ for $d$ defined in \autoref{aspace} such that for all $r\geq 0$,
  \begin{equation*}
    \sup_{n\in\mathbb{N}} \max_{i\in\N_n} \max\big\{\abs{Y_i(\bm{d}) - Y_i(\bm{d}')}\colon \bm{d},\bm{d}'\in\{0,1\}^n, d_j=d_j' \,\,\forall j\in\N(i,r)\big\} \leq c\,\min\{r^{-\gamma},1\}. 
  \end{equation*}
\end{assump}

\noindent This corresponds to Assumption 3 of \cite{leung2022rate}. Unlike partial interference, it is formulated independently of clusters, enabling us to develop a theory of optimal cluster construction. 
To interpret the condition, consider a unit $i$ centered at a ``cluster'' $\N(i,r)$ of radius $r$. The quantity $\abs{Y_i(\bm{d}) - Y_i(\bm{d}')}$ measures interference induced by manipulating the treatment assignments of units outside the cluster, and the assumption requires this to decay like $r^{-\gamma}$ or faster. Hence, the larger the minimum distance $r$ between $i$ and the units with manipulated treatments, the smaller the spillover effect. We thus interpret $1/\gamma$ as (an upper bound on) the {\em degree of interference}. In \autoref{sopt}, we discuss the optimal design of clusters using knowledge of $\gamma$.

\autoref{aani} further requires $\gamma$ to decay fast enough relative to the spatial dimension $d$. This coincides with the requirement imposed by \cite{leung2022rate} for $d=2$. In the increasing-domain case, the intuition is that $r$-neighborhood sizes grow like $r^d$ under \autoref{aspace}, so weak dependence requires interference to decay faster than the rate at which neighborhoods densify with $r$. Central limit theorems for spatial processes impose analogous requirements on mixing coefficients which control the degree of spatial dependence, as $\gamma$ does in our framework \citep[e.g.][Assumption 3(b)]{jenish2009central}. 

%---------------------------------------
\subsection{Design}\label{sdesign}
%---------------------------------------

We interchangeably use $k$ or $k_n$ to denote the number of clusters in the design. For a given population $\N_n$, denote by $\C_n = \{C_j\}_{j=1}^k$ the set of clusters, which is a partition of $\N_n$. We study standard two-stage randomized-saturation designs.

\begin{assump}[Assignment Mechanism]\label{aCRT}
  For $q \in (0,1)$, $p_0, p_1 \in [0,1]$, and $t \in \{0,1\}$,  $\{D_i\}_{i \in C_j} \stackrel{iid}\sim \text{Bernoulli}(p_t)$ conditional on $W_j=t$ where $\{W_j\}_{j=1}^k \stackrel{iid}\sim \text{Bernoulli}(q)$.
\end{assump}

\noindent Under this design, $k$ clusters are independently randomized into treatment with probability $q$ with $W_j$ denoting cluster $j$'s treatment assignment. If $W_j=t$, all units in cluster $j$ are randomized into treatment with probability $p_t$. The literature often refers to $p_1,p_0$ as {\em saturation levels}.

Let $\xi_n$ and $d$ be given from \autoref{aspace}. We consider (sequences of) clusters satisfying the following properties.

\begin{assump}[Spatial Clusters]\label{aclus}
  There exist positive sequences $\{L_n\}_{n\in\mathbb{N}}$ and $\{U_n\}_{n\in\mathbb{N}}$ for which the following hold. (a) For any sequence of clusters $\{C_n\}_{n\in\mathbb{N}}$ with $C_n \in \C_n$ for all $n$, there exists a sequence of ``centroid'' units $\{i_n\}_{n\in\mathbb{N}}$ with $i_n \in C_n$ for all $n$ such that $\N(i_n,L_n) \subseteq C_n \subseteq \N(i_n,U_n)$. (b) $L_n \sim U_n \sim (n/(k_n\xi_n))^{1/d}$. 
\end{assump}

\noindent This requires clusters to be globular in that they contain and are contained by balls with radii of the same order. Under \autoref{aspace}(a), part (b) implies that the number of units in any cluster is order $n/k_n$, which implies that clusters are comparable in size. We show below that clusters generated by $k$-medoids satisfy these requirements. Note that even under partial interference, restrictions on cluster size heterogeneity are required for inference using conventional clustered standard errors \citep[][Assumption 2]{hansen2019asymptotic}, although our requirements are stronger. These may be possible to relax, but we leave this to future research.

\autoref{aclus} allows clusters to have quite heterogeneous shapes, as can be seen in \autoref{fclusters}. It depicts $k$-medoid clusters, which satisfy \autoref{aclus} by \autoref{tkmed} below. However, the assumption is violated if clusters are highly size-imbalanced, such as elongated clusters that narrowly encompass a geographical feature such as a river or road. This may be addressed by subdividing large or abnormally shaped clusters or grouping adjacent small clusters. The assumption can also be violated if clusters are not constructed based on spatial proximity. For instance, if units are connected through an online social network connecting units far apart in space, then clusters based on network connectivity would likely violate the assumption. 

\autoref{aclus} is satisfied if the researcher partitions the space into $k_n$ identically-sized cubes, but such clusters do not adapt to the spatial distribution of units. We therefore suggest using unsupervised learning algorithms. The next result provides theoretical guarantees for the well-known $k$-medoids algorithm stated in \autoref{skmed}.

\begin{theorem}\label{tkmed}
  Suppose clusters are the output of $k$-medoids, given in \autoref{amedoids}. Under \autoref{aspace}(a), the clusters satisfy \autoref{aclus}.
\end{theorem}

In practice, $k$-means can also be used since it typically delivers clusters similar to those of $k$-medoids. The globular clusters produced by both algorithms are sometimes viewed unfavorably relative to algorithms such as spectral clustering for certain unsupervised learning tasks. However, for our purposes, globular clusters are preferable because they minimize the number of units near cluster boundaries, which are the primary source of bias due to cross-cluster interference. More broadly, \autoref{aclus} provides general design principles for constructing clusters, namely to aim for balance and globularity, which may be achieved either using these algorithms or manually.

%---------------------------------------
\subsection{Causal Estimands}\label{sestimands}
%---------------------------------------

Because we allow for cross-cluster interference, we need to redefine conventional estimands in a manner free of this source of bias. To this end, define for $t \in \{0,1\}$ the $p_t$-{\em counterfactual design}, which sets $k=1$ and $q=t$ in \autoref{aCRT}. This design groups the population into a single cluster and assigns treatment with probability $p_t$. Throughout the paper, let $\E[\cdot]$ denote the expectation taken with respect to the observed design in \autoref{aCRT} and $\E^*_{p_t}[\cdot]$ the expectation taken with respect to the $p_t$-counterfactual design. 

Let $\bm{D}_{-i}$ denote the assignment vector excluding the $i$th component and $Y_i(d, \bm{D}_{-i})$ denote $i$'s potential outcome under the counterfactual that $i$'s observed assignment is $d \in \{0,1\}$, holding fixed the realized assignments of other units. We study estimands of the form
\begin{equation*}
  \theta^* \equiv \theta^*(d_1, d_0; p_1, p_0) = \frac{1}{n} \sum_{i\in\N_n} \big( \E^*_{p_1}[Y_i(d_1,\bm{D}_{-i})] - \E^*_{p_0}[Y_i(d_0,\bm{D}_{-i})] \big) 
\end{equation*}

\noindent for $d_1,d_0 \in \{\emptyset,0,1\}$, where we define $Y_i(\emptyset,\bm{D}_{-i}) \equiv Y_i(\bm{D})$. The following special cases are analogous to estimands defined by \cite{hudgens2008toward} and \cite{hayes2017cluster}. The {\em direct effect} compares treated and untreated units under the $p_1$-counterfactual design:
\begin{equation*}
  \theta_D^* = \theta^*(1, 0, p_1, p_1) = \frac{1}{n} \sum_{i\in\N_n} \E^*_{p_1}[Y_i(1,\bm{D}_{-i}) - Y_i(0,\bm{D}_{-i})].
\end{equation*}

\noindent This is directly identified if the $p_1$-counterfactual design were implemented in practice. The estimands that follow, however, require multiple clusters assigned to different arms because they compare different counterfactual designs. Under our framework, this is the primary motivation for cluster randomization.

The {\em indirect effect} compares outcomes of untreated units under counterfactual designs with different saturation levels:
\begin{equation*}
  \theta_I^* = \theta^*(0, 0; p_1, p_0) = \frac{1}{n} \sum_{i\in\N_n} \big( \E^*_{p_1}[Y_i(0,\bm{D}_{-i})] - \E^*_{p_0}[Y_i(0,\bm{D}_{-i})] \big).
\end{equation*}

\noindent Interest often centers on the ``pure control'' baseline of $p_0=0$, so that $\E^*_{p_0}[Y_i(0,\bm{D}_{-i})] = Y_i(\bm{0})$. The {\em total effect} is the sum of the direct and indirect effects, equal to $\theta_T^* = \theta^*(1, 0; p_1, p_0)$. Finally, the {\em overall effect} compares outcomes under different saturation levels: $\theta_O^* = \theta^*(\emptyset, \emptyset, p_1, p_0)$. \cite{leung2024causal} provides conditions under which these have causal interpretations. 

%----------------------------------------------------------------------
\section{Estimators}\label{sests}
%----------------------------------------------------------------------

%---------------------------------------
\subsection{Existing Approaches}\label{sprior}
%---------------------------------------

Let $c(i)\in\{1,\ldots,k\}$ denote the index of the cluster containing unit $i$. A common strategy for estimating $\theta_I^*$ is to compute the difference in means between control units in clusters assigned saturation level $p_1$ and those in clusters assigned level $p_0$:
\begin{equation}
  \frac{\sum_{i\in\N_n} (1-D_i)W_{c(i)}Y_i}{\sum_{i\in\N_n} (1-D_i)W_{c(i)}} - \frac{\sum_{i\in\N_n} (1-W_{c(i)})Y_i}{\sum_{i\in\N_n} (1-W_{c(i)})}. \label{naive}
\end{equation}

\noindent This is the sample analog of
\begin{equation*}
  \frac{1}{n} \sum_{i\in\N_n} \big(\E[Y_i \mid (1-D_i)W_{c(i)}=1] - \E[Y_i \mid W_{c(i)}=0] \big),
\end{equation*}

\noindent which may be quite far from the target $\theta_I^*$. Units in control clusters near cluster boundaries may be spatially proximate to units in treated clusters and therefore at greater risk of contamination. For such units $i$, $\E[Y_i \mid W_{c(i)}=0]$ may be substantially different from $\E^*_{p_0}[Y_i(0,\bm{D}_{-i})]$. 

The fried-egg design attempts to reduce bias by restricting the comparison in \eqref{naive} to the subset of units deemed sufficiently far from cluster boundaries. Unfortunately, this has two problems. First, the resulting estimator is in fact asymptotically biased because boundary units are excluded with probability one, so it only estimates an average effect for the subpopulation of units in cluster interiors. As noted by \cite{mccann2018reducing}, this differs from the target estimand since units in cluster interiors may be systematically different from those near the boundaries due to spatial heterogeneity. Second, it is inefficient. The purpose of only including units in the interiors of control clusters is that such units are ``well surrounded'' by control clusters, or equivalently, relatively far from treated clusters. However, boundary units may be well surrounded in the same fashion if all neighboring clusters are assigned to control, so it would be just as useful to include these units.

%---------------------------------------
\subsection{Our Approach}\label{sest}
%---------------------------------------

Fix a neighborhood radius $r_n$ to be defined in \eqref{kappa} below. Call a unit $i$ {\em well-surrounded} if 
\begin{equation*}
  S_i \equiv \max_{t \in \{0,1\}} \prod_{j \in \N(i,r_n)} W_{c(j)}^t(1-W_{c(j)})^{1-t} = 1,
\end{equation*}

\noindent that is, if a unit $i$'s $r_n$-neighborhood only intersects clusters assigned to the same treatment arm. \autoref{fclusters} depicts $k$-medoid clusters, marking with an ``X'' units that are not well surrounded. Our strategy is to only exclude from estimation units that are not well surrounded. If $r_n=0$, then all units are well surrounded, and our estimators reduce to difference in means. Choosing a larger radius $r_n$ is analogous to choosing a larger boundary region for exclusion in a fried-egg design. 

\begin{figure}
  \centering
  \includegraphics[scale=0.55]{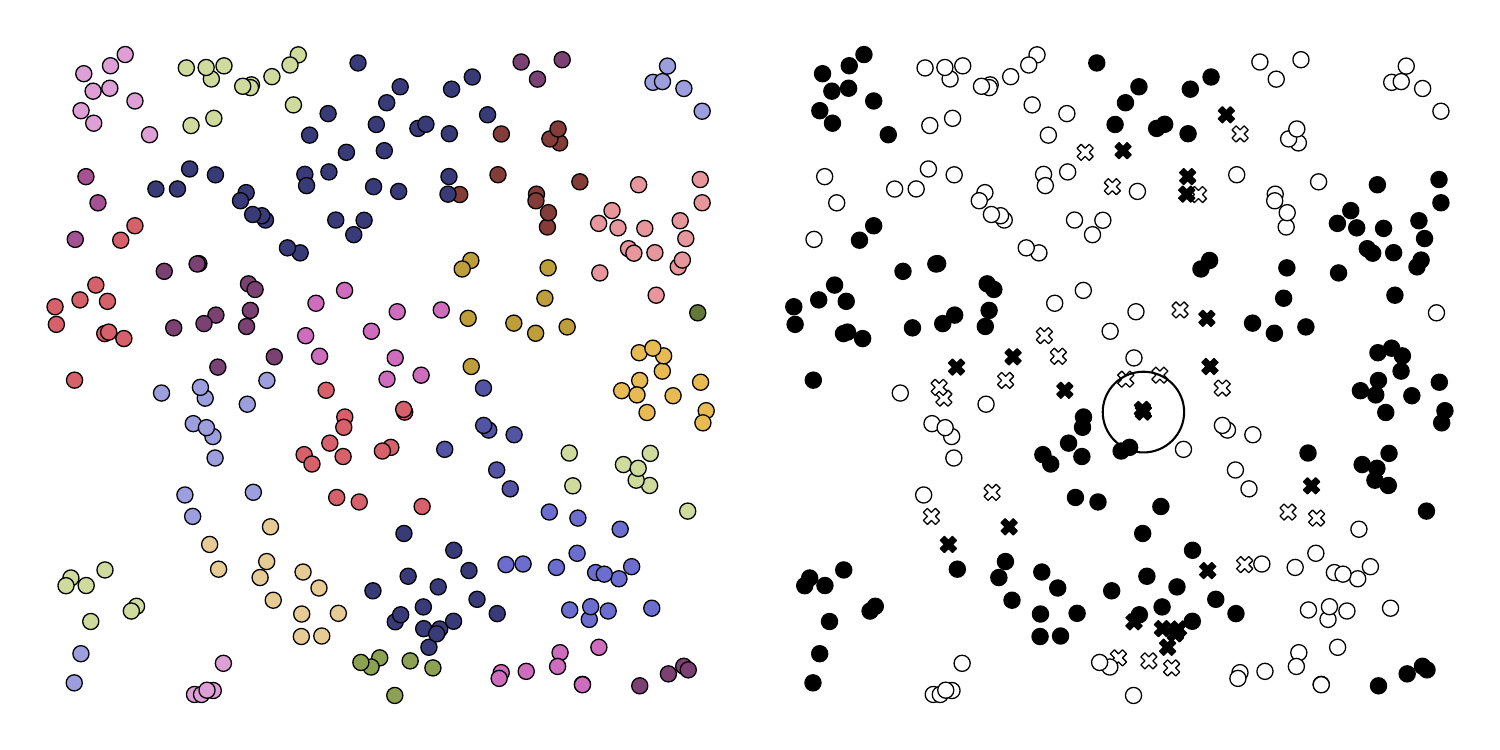}
  \caption{\small $k$-medoid clusters. The left panel colors units by cluster membership and the right panel by whether their respective clusters are assigned to treatment (black) or control (white). Units marked by ``X'' are excluded from estimation. The circle depicts a unit's $r_n$-neighborhood used to determine exclusion.}
  \label{fclusters}
\end{figure}

Under our assumptions, the number of well-surrounded units is of asymptotic order equal to the population size $n$ because any unit has nontrivial probability of being well-surrounded (\autoref{lintersect}). As a consequence, our strategy solves the fried-egg design's boundary bias problem and typically excludes strictly fewer units.

Let $d_1,d_0$ be given from the estimand $\theta^*$. For any $t \in \{0,1\}$ and $i \in \N_n$, define $T_{ti} = \ind\{D_i=d_t, W_{c(i)}=t\}S_i$ if $d_t \in \{0,1\}$ and $T_{ti} = \ind\{W_{c(i)}=t\}S_i$ if $d_t = \emptyset$. Let $p_{ti} = \E[T_{ti}]$ be the propensity score, which has a closed-form expression given in \autoref{roverlap} below. We propose the following H\'{a}jek estimator for $\theta^*$:
\begin{equation*}
  \hat\theta = \hat\mu_1 - \hat\mu_0 \quad\text{for}\quad \hat\mu_t = \frac{\sum_{i\in\N_n} T_{ti} Y_i / p_{ti}}{\sum_{i\in\N_n} T_{ti} / p_{ti}}.
\end{equation*}

\noindent When $r_n=0$, $S_i=1$ for all $i$, so propensity scores are homogeneous across $i$, and $\hat\theta$ reduces to difference in means. For $r_n>0$, the scores are generally spatially heterogeneous. For instance, boundary units are less likely to be well surrounded than interior units because that would require more clusters to be assigned the same saturation level. 

For any cluster $C_j$, let $m_j$ be its centroid from \autoref{aclus} and $R_j = \max_{i \in C_j} \rho(i,m_j)$, the ``radius'' of the cluster. We propose setting
\begin{equation}
  r_n = 0.5 \cdot \text{Median}(\{R_j\}_{j=1}^k). \label{kappa}
\end{equation}

\noindent In the case where clusters are equally sized squares, this coincides with the radius suggested by \cite{leung2022rate}. 

\begin{remark}[Overlap]\label{roverlap}
  Let $\phi_i$ be the number of clusters intersecting $\N(i,r_n)$. The propensity score can be written explicitly as
  \begin{equation*}
    p_{ti} = p_t^{d_t} (1-p_t)^{1-d_t} q^{\phi_i t} (1-q)^{\phi_i(1-t)} 
  \end{equation*}

  \noindent where $p_t^\emptyset \equiv 1$. We show in \autoref{lintersect} that $\phi_i$ is asymptotically bounded uniformly in $i$ under \eqref{kappa}. Since $q \in (0,1)$ by \autoref{aCRT}, $p_{ti}$ is similarly bounded away from zero when the saturation levels $p_1,p_0$ are nontrivial (see \autoref{aoverlap} below). Hence, using a CRT paired with our estimation strategy inherently avoids positivity violations or limited overlap.
\end{remark}

\begin{remark}
  The choice of 0.5 in \eqref{kappa} is immaterial for the asymptotic theory, but in finite samples it controls a bias-variance trade-off. We choose 0.5 to informally balance the two. Constants close to zero will yield high bias because $r_n=0$ corresponds to difference in means. Constants close to one result in high variance. To see this, suppose each cluster is a ball with homogeneous radius $R$. Setting the constant to 1 implies $r_n = R$, so the only unit in any cluster that is well surrounded with probability one is the centroid. A unit located slightly north of the centroid finds its $r_n$-neighborhood intersecting at least two clusters, so the chance that it is well surrounded, and hence not excluded, is substantially lower. In contrast, suppose $r_n = 0.5R$, and define a cluster $C_j$'s ``interior'' as the ball centered at $m_j$ with radius $R/2$. Then all units in the interior are well surrounded with probability one because their $r_n$-neighborhoods are contained within $C_j$. 
\end{remark}

Lastly, we propose a variance estimator for $\hat\theta$. Let
\begin{equation}
  \Lambda_i = \big\{ \ell\in\N_n\colon \max_{j\in\{1,\ldots,k_n\}} \, \abs{C_j \cap \N(i,r_n)} \, \abs{C_j \cap \N(\ell,r_n)} > 0 \big\},
  \label{Lambdai}
\end{equation}

\noindent the set of units $\ell$ for which some cluster intersects the $r_n$-neighborhood of $\ell$ and $i$. These can be thought of as the units ``most potentially correlated'' with $i$. Define $A_{ij}(1) = \bm{1}\{j \in \Lambda_i\}$, $A_{ij}(2) = \bm{1}\{j \in C_{c(i)}\}$, and $\hat{Z}_i = (T_{1i}(Y_i - \hat\mu_1))/p_{1i} - (T_{0i}(Y_i - \hat\mu_0))/p_{0i}$. The variance estimator is
\begin{equation}
  \hat\sigma^2 = \max_{t \in \{1,2\}} \hat\sigma^2(u) \quad\text{where}\quad \hat\sigma^2(u) = \frac{k_n}{n^2} \sum_{i\in\N_n} \sum_{j\in\N_n} \hat{Z}_i \hat{Z}_j A_{ij}(u). \label{hatsigma}
\end{equation}

\noindent Notice that $\hat\sigma^2(2)$ is the conventional cluster-robust variance estimator \citep[e.g.][]{baird2018optimal}, which only accounts for within-cluster dependence. The estimator $\hat\sigma^2(1)$ is analogous to that of \cite{leung2022rate} and additionally accounts for cross-cluster dependence since $C_{c(i)} \subseteq \Lambda_i$. In \autoref{sasymp}, we discuss the advantages of taking the larger of the two. 

%----------------------------------------------------------------------
\section{Asymptotic Theory}\label{sasymp}
%----------------------------------------------------------------------

We next derive bounds on the rates of convergence of our estimator and difference in means. We then characterize the asymptotic distribution of our estimator and prove that the variance estimator is asymptotically conservative. 

\begin{assump}[Overlap]\label{aoverlap}
  For all $t \in \{0,1\}$, if $d_t=1$ ($d_t=0$), then $p_t > 0$ ($p_t < 1$).
\end{assump}

\noindent As discussed in \autoref{roverlap}, this ensures overlap or positivity. For the next two theorems, abbreviate the estimand $\theta^*(d_1,d_0,p_1,p_0)$ as $\theta^*$. Recall the asymptotic order notation from the end of \autoref{sintro}. 

\begin{theorem}[Our Estimator]\label{tQ}
  Suppose $k_n \precsim n/\xi_n$. Under Assumptions \ref{aboundY}--\ref{aoverlap}, $\abs{\hat\theta - \theta^*} \precsim r_n^{-\gamma} + k_n^{-1/2} \sim (k_n\xi_n/n)^{\gamma/d} + k_n^{-1/2}$.
\end{theorem} 

\noindent The result establishes a bias-variance trade-off in $k_n$. The $k_n^{-1/2}$ term is the contribution of (the square root of) the variance since the effective sample size in a CRT is the number of clusters $k_n$. The asymptotic bias is order $r_n^{-\gamma}$. Intuitively, if $r_n$ is large, then the expected outcome of a unit $i$ assigned treatment $d_t$ and well surrounded by clusters assigned to saturation level $p_t$ should well approximate $\E^*_{p_t}[Y_i(d_t,\bm{D}_{-i})]$, corresponding to lower bias. The bias decreases at a faster rate for larger $\gamma$ since this corresponds to a lower degree of spatial interference.

The requirement $k_n \precsim n/\xi_n$ is mild, and typically we would have $k_n \prec n/\xi_n$. In the increasing-domain case where $\xi_n=1$, necessarily $k_n \precsim n/\xi_n$ since $k_n \leq n$, and usually the number of clusters is of smaller order than the population size. In the infill-increasing case, if $k_n$ grows faster than $n/\xi_n$, which is the volume of the study region under \autoref{aspace}, then we would be increasingly subdividing the region into smaller clusters of shrinking volume, analogous to having more clusters than units. 

Denote the difference-in-means estimator of $\theta^*$ by $\hat\theta^+$, which corresponds to setting $r_n=0$ in the definition of $\hat\theta$. Let $\hat\theta_D^+$, $\hat\theta_T^+$, and $\hat\theta_O^+$ be the difference-in-means estimates of $\theta_D^*$, $\theta_T^*$, and $\theta_O^*$ defined in \autoref{sestimands}.

\begin{theorem}[Difference in Means]\label{tplus} 
  (a) Under Assumptions \ref{aboundY}--\ref{aoverlap}, $\abs{\hat\theta^+ - \theta^*} \precsim (k_n\xi_n/n)^{1/d} + k_n^{-1/2}$. (b) Suppose the design satisfies \autoref{aCRT} for $p_1 \in (0,1)$, $p_0=0$. There exist a metric space and sequence of units, clusters, and potential outcomes satisfying Assumptions \ref{aboundY}--\ref{aani} and \ref{aclus} such that $\abs{\hat\theta_Q^+ - \theta_Q^*} \succsim (k_n\xi_n/n)^{1/d} + k_n^{-1/2}$ for any $Q \in \{D,T,O\}$.
\end{theorem}

\noindent Part (a) provides an upper bound on the rate. Part (b) shows that the rate is tight for several illustrative cases. Recall that setting $p_0=0$ corresponds to the ``pure control'' baseline used in the estimands of \cite{hayes2017cluster}.

The $k_n^{-1/2}$ term in the bound is (the square root of) the variance contribution, as in \autoref{tQ}, while $(k_n\xi_n/n)^{1/d}$ is the bias contribution. The latter is notably worse than that of $\hat\theta$ because a reduction in the degree of interference $1/\gamma$ has no effect on the rate. The reason is that units situated at cluster boundaries may be directly proximate to clusters assigned to different treatment arms, and as shown in the proof, the share of such units can be of order $(k_n\xi_n/n)^{1/d}$. 

These results demonstrate that excluding units in the manner of $\hat\theta$ can significantly reduce the asymptotic order of the bias relative to difference in means. While excluding units may come at the cost of efficiency, there is no increase to the asymptotic order of the variance because the variances of both estimators scale not with the number of units but with the number of clusters. Hence, the efficiency loss is second-order relative to the potential reduction in bias.

Define $\mu_t = n^{-1} \sum_{i\in\N_n} \E[Y_i \mid T_{ti}=1]$, $\bar{\theta} = \mu_1 - \mu_0$, and
\begin{equation*}
  \sigma_n^2 = \var\bigg( \sqrt{k_n} \frac{1}{n} \sum_{i\in\N_n} \bigg( \frac{T_{1i}(Y_i - \mu_1)}{p_{1i}} - \frac{T_{0i}(Y_i - \mu_0)}{p_{0i}} \bigg) \bigg).
\end{equation*}

\noindent Note that $\sigma_n^2 \precsim 1$ by the proof of \autoref{tQ}.

\begin{theorem}[CLT]\label{tclt}
  Suppose $1 \prec k_n \prec n/\xi_n$ and $\sigma_n^2 \succsim 1$. Under Assumptions \ref{aboundY}--\ref{aoverlap}, 
  \begin{equation}
    \sigma_n^{-1} \sqrt{k_n} (\hat\theta - \bar{\theta}) \dlimarrow \N(0,1). \label{cltsec}
  \end{equation}

  \noindent Furthermore, if $k_n \prec (n/\xi_n)^{\frac{2\gamma}{2\gamma+d}}$, then
  \begin{equation}
    \sigma_n^{-1} \sqrt{k_n} (\hat\theta - \theta^*) \dlimarrow \N(0,1). \label{cltmain}
  \end{equation}
\end{theorem} 

\noindent The first result \eqref{cltsec} centers the estimator at $\bar{\theta}$, the probability limit of $\hat\theta$. This is not the estimand of interest since there is an asymptotic bias $\abs{\bar{\theta} - \theta^*}$ by Theorem \ref{tQ}. To use the normal limit to justify the validity of conventional CIs for $\theta^*$, the second result \eqref{cltmain} requires ``undersmoothed designs'' in which the number of clusters is of smaller order than the optimal rate discussed in the next section. This ensures that the asymptotic bias is small. It is analogous to nonparametric regression where rate-optimal tuning parameter choices result in asymptotic bias, so conventional CIs require undersmoothing. We discuss practical choices of $k_n$ in \autoref{sopt}.  In \autoref{sBA}, we discuss ``bias-aware'' inference, an alternative to undersmoothing.

\begin{theorem}[Variance Estimator]\label{tvar}
  Suppose $1 \prec k_n \prec n/\xi_n$ and $\sigma_n^2 \succsim 1$. Under Assumptions \ref{aboundY}--\ref{aoverlap}, $\hat\sigma^2(t) = \sigma_n^2 + \mathcal{B}_n + o_p(1)$ for $t \in \{1,2\}$ and some sequence of non-negative constants $\{\mathcal{B}_n\}_{n\in\mathbb{N}}$. Hence, both $\hat\sigma^2(1)$ and $\hat\sigma^2(2)$ are asymptotically conservative.
\end{theorem} 

\noindent The cluster-robust variance estimator $\hat\sigma^2(2)$ has the advantage of being non-negative in finite sample, unlike $\hat\sigma^2(1)$ which is a truncation estimator \citep[][p.\ 823]{andrews1991heteroskedasticity}. On the other hand, $\hat\sigma^2(2)$ only accounts for within-cluster dependence, whereas $\hat\sigma^2(1)$ also accounts for cross-cluster dependence in the definition $A_{ij}(1)$. As shown in the proof, this dependence vanishes, so $\hat\sigma^2(1)$ is a valid estimator. However, the dependence vanishes at the slow rate $(k_n\xi_n/n)^{1/d}$, the same order as the bias of difference in means. Thus in smaller samples, ignoring second-order terms may result in anti-conservativeness. By taking the larger of the two estimators, we obtain the benefits of both. See \autoref{svar} for a comparison of $\hat\sigma^2$ with other variance estimators in the literature.

%----------------------------------------------------------------------
\section{Optimal Design}\label{sopt}
%----------------------------------------------------------------------

Recall that $d$ is the dimension of the spatial region, $\xi_n$ is the density of the region (number of units per unit volume/area), and $\gamma$ is a lower bound on the speed at which interference decays with distance. By \autoref{tQ}, choosing
\begin{equation}
  k \sim (n/\xi_n)^{\frac{2\gamma}{2\gamma+d}} \label{rok}
\end{equation}

\noindent optimizes the rate of convergence of $\hat\theta$. This formalizes how domain knowledge of interference $\gamma$ informs the design of clusters. The right-hand side is increasing in $\gamma$ since less interference means the same level of bias reduction can be achieved with more clusters. It is decreasing in density $\xi_n$ because having more units in a given area effectively corresponds to greater interference (an infectious disease may spread more easily).

Choosing the number of clusters according to \eqref{rok} balances the asymptotic orders of the bias and variance of $\hat\theta$. However, as discussed in \autoref{sasymp}, validity of the CI 
\begin{equation}
  \hat\theta \pm 1.96 \cdot \hat\sigma k^{-1/2} \label{undersmoothCI}
\end{equation}

\noindent requires the bias to be of smaller order than the variance. This requires an {\em undersmoothed design} where $k$ is of smaller order than \eqref{rok}.

We next recommend an undersmoothed choice of $k$ given domain knowledge of $\gamma$. Let $\mathcal{V}$ denote the area or volume of the study region containing $\N_n$. As discussed below, $\mathcal{V}$ is of order $n/\xi_n$, so the rate-optimal formula \eqref{rok} can be rewritten as $k \sim \mathcal{V}^{\frac{2\gamma}{2\gamma+d}}$. While one could choose $k$ equal to the right-hand side, this choice would be extremely sensitive to the unit of length used to measure geographic distance. Switching from square kilometers to square millimeters would dramatically increase $\mathcal{V}$. 

Our observation is that both the unit of length and $\gamma$ determine the speed at which interference decays with distance. For any given choice of $\gamma$, the rate of decay is significantly faster if distance is measured in millimeters compared to kilometers, for example. Therefore, domain knowledge of interference informs both $\gamma$ and the unit of length, and one can determine $k$ from these ingredients as follows.

\begin{enumerate}
  \item Based on domain knowledge, specify a unit of length and rate of decay $\gamma$ under which \autoref{aani} is believed to hold. Choose a strict lower bound $\tilde\gamma < \gamma$. For a given unit of length, the conservative choice allowed by \autoref{aani} is $\tilde\gamma=d$. We suggest this choice absent any prior information on $\gamma$.

  \item Let $\mathcal{V}$ denote the volume of the study region containing $\N_n$ under the chosen unit of length. Set the number of clusters to
    \begin{equation}
      k = \big[\min\{\mathcal{V},n\}^{\frac{2\tilde\gamma}{2\tilde\gamma + d}}\big],
      \label{optk}
    \end{equation}

    \noindent where $[c]$ means round $c$ to the nearest integer. 
\end{enumerate}

The next subsection provides empirical examples of calibrating $k$ under the conservative choice $\tilde\gamma=d$. The last subsection discusses how to determine $\gamma$ to obtain less conservative estimates. 

\begin{remark}
  The theoretical motivation for \eqref{optk} is as follows. First, using $\tilde\gamma < \gamma$ corresponds to undersmoothing. Second, \autoref{aspace}(a) says that the number of units in an $r$-neighborhood is of asymptotic order equal to the density $\xi_n$ times the volume of the neighborhood $r^d$. It follows that the number of units $n$ is of asymptotic order equal to $\xi_n$ times the volume of the region: $n \sim \xi_n \mathcal{V}$. Hence, in the infill-increasing case where $\xi_n$ is diverging, we replace $n/\xi_n$ in the rate-optimal formula \eqref{rok} with $\mathcal{V}$, which can be directly computed from the data. In the increasing-domain case where $\xi_n \sim 1$, we can replace $n/\xi_n$ with either $n$ or $\mathcal{V}$ since they are of the same order, so \eqref{optk} conservatively chooses the smaller option. 
\end{remark}

%---------------------------------------
\subsection{Empirical Examples}
%---------------------------------------

\cite{sur2009cluster} conduct a CRT in an urban slum in India spanning about 1.2 by 0.7 km with 38k participants. To compute our suggested number of clusters \eqref{optk}, we need to select $\tilde\gamma$ and the unit of length. Suppose we choose the conservative bound $\tilde\gamma = d = 2$, so that interference is assumed to decay like $r^{-2}$ or faster with each unit of length $r$. If we take 35m to be the unit of length, then for units $i,j,k$ such that $\rho(i,j) = 35$m and $\rho(i,k) = 70$m, the extent to which $k$'s treatment affects $i$ is less than $1/4$th ($2^{-\tilde\gamma}=0.25$) as small as the extent to which $j$'s affects $i$. This is with only a 35m difference in distance. For this unit of length, \eqref{optk} yields $k = 78$. In other words, these are the assumptions on interference that justify the authors' choice of $k=80$, originally determined by a conventional power analysis based on a partial interference model.

The previous rate of decay may be over-optimistic, so suppose the relevant unit of length is in 100m increments. If $\rho(i,j)=100$m and $\rho(i,k) = 200$m, the extent to which $k$'s treatment affects $i$ is less than $1/4$th as small as the extent to which $j$'s affects $i$, now with a 100m difference in distance. With a slower rate of decay, bias is higher, which requires constructing fewer, larger clusters. As a result, \eqref{optk} yields only $k = 19$. 

Next consider \cite{homan2016effect} whose trial region is substantially larger at roughly 12 by 4 km with nearly the same number of units (34k). Due to the lower density, there is less bias from interference, so a choice of $k=81$ can be justified under weaker assumptions on interference. Specifically, if $\tilde\gamma=2$ but the unit of length is now in 250m increments, then \eqref{optk} results in $k=84$. For context, the {\em Anopheles} mosquito that is the subject of their trial typically does not fly more than 2 km from their breeding grounds \citep{anopheles}.

These examples illustrate what domain knowledge of the degree of interference entails. They also show how our proposed choice of $k$ accounts for regional density, unlike the standard power analysis \citep[e.g.][eq.\ (9)]{hemming2011sample}. In both examples, weaker assumptions on interference require choosing $k$ smaller than the standard analysis to better balance bias and variance. To justify choosing larger $k$, the researcher must either collect data over a larger area or be willing to entertain stronger assumptions on interference.

%---------------------------------------
\subsection{Bounding Interference}\label{sgamma}
%---------------------------------------

Optimal design generally requires prior information on certain population parameters. The standard power analysis for CRTs assumes partial interference and requires knowledge of the intracluster correlation coefficient, a measure of within-cluster dependence in potential outcomes \citep[e.g.][]{baird2018optimal,hemming2017design}. Our results suggest that, if cross-cluster interference is of first-order importance, the focus of attention should instead be $\gamma$. As illustrated in the previous subsection, the conservative choice $\gamma=d$ can result in relatively small values of $k$ because it allows for a greater degree of interference, so to the extent that one can justify stronger assumptions on interference, that is, values of $\gamma$ that are larger than $d$ for a given unit of length, this would improve asymptotic power.

In the context of infectious disease trials, \cite{halloran2017simulations} argue that CRT design should be informed by simulating models of disease transmission. Several papers utilize parametric models and simulation methods to estimate or bound contamination bias. \cite{alexander2020spatial} and \cite{jarvis2019spatial} use spatial models to provide evidence of contamination in prior CRTs. \cite{multerer2021analysis} employ models of disease transmission for a similar purpose. Our theory provides a precise way in which modeling can inform design, namely by providing plausible values of $\gamma$. This relates to \cite{staples2015incorporating} who show how to estimate a different measure of cross-cluster interference to assess the degree to which the conventional power analysis overstates trial power. 

To estimate $\gamma$, models may be combined with external data sources such as data from pilot studies. In the context of malaria vector control, mosquito mark-release-recapture experiments \citep[e.g.][]{guerra2014global} provide data on geographic dispersion of malaria carriers, which is informative of spatial interference. We provide additional suggestions in \autoref{sgammax}.

%----------------------------------------------------------------------
\section{Simulation Study}\label{smc}
%----------------------------------------------------------------------

We conduct a simulation study to illustrate the finite-sample properties of our estimator and difference in means under our proposed design. We randomly draw unit locations from the square $[-(n\alpha_n)^{1/2}, (n\alpha_n)^{1/2}]^2$ with $\alpha_n = 0.8, 0.7, 0.6$, respectively. This corresponds to the infill-increasing case where the regional volume shrinks with the population size. We create clusters using $k$-medoids with $k$ given by \eqref{optk} using the conservative choice $\tilde\gamma=d=2$ and set the assignment probabilities in \autoref{aCRT} to $(q,p_1,p_0) = (0.7,0.5,0)$.

Let $\{\tilde\varepsilon_i\}_{i\in\N_n} \stackrel{iid}\sim \mathcal{N}(-0.5,1)$, $\{\beta_i\}_{i \in \N_n} \stackrel{iid} \sim \mathcal{N}(2,1)$, and $\{\gamma_i\}_{i \in \N_n} \stackrel{iid} \sim \mathcal{N}(1,1)$ be independent and drawn independently of locations. We generate spatially autocorrelated errors $\varepsilon_i = \tilde\varepsilon_i + \sum_{j\in\N_n} \bm{1}\{\rho(i,j) \leq 1\}\tilde\varepsilon_j / \sum_{k\in\N_n} \bm{1}\{\rho(i,k) \leq 1\}$. For $w_{ij} = \min\{\rho(i,j)^{-5}, 1\}$, we generate outcomes according to
\begin{equation*}
  Y_i = \sum_{j\in\mathcal{N}_n} w_{ij}D_j\beta_j + \sum_{j\in\N_n} w_{ij} D_iD_j\gamma_j + \varepsilon_i.
\end{equation*}

\noindent Under this model, the unit-level direct and indirect effects are respectively given by
\begin{equation*}
  Y_i(1,\bm{D}_{-i}) - Y_i(0,\bm{D}_{-i}) = \gamma_i+\beta_i + \sum_{j\neq i} w_{ij} D_j\gamma_j \quad\text{and}\quad Y_i(0,\bm{D}_{-i}) - Y_i(\bm{0}) = \sum_{j\neq i} w_{ij} D_j\beta_j.
\end{equation*}

\noindent Due to the choice of $-5$ in the spatial weights $w_{ij}$, \autoref{aani} holds for $\gamma = 3$ \citep[][eq.\ (3)]{leung2022rate}, which is a fairly slow rate of decay given that \autoref{aani} requires $\gamma>2$. 

\begin{table}[ht]
\centering
\caption{Main results}
\resizebox{\columnwidth}{!}{%
\begin{tabular}{lrrrrrrrrrrrr}
\toprule
 & \multicolumn{6}{c}{Spatial Interference} & \multicolumn{6}{c}{Partial Interference} \\
 \cmidrule(lr){2-7} \cmidrule(lr){8-13}
 & \multicolumn{3}{c}{Indirect Effect} & \multicolumn{3}{c}{Overall Effect} & \multicolumn{3}{c}{Indirect Effect} & \multicolumn{3}{c}{Overall Effect} \\
 \cmidrule(lr){2-4} \cmidrule(lr){5-7} \cmidrule(lr){8-10} \cmidrule(lr){11-13}
$n$ & 500 & 1000 & 2000 & 500 & 1000 & 2000 & 500 & 1000 & 2000 & 500 & 1000 & 2000 \\
\midrule
\multicolumn{2}{l}{Our Estimator} \\
Bias & 0.064 & 0.055 & 0.051 & 0.072 & 0.059 & 0.057 & 0.000 & 0.003 & 0.005 & 0.003 & 0.004 & 0.004 \\
CI SE & 0.945 & 0.951 & 0.954 & 0.940 & 0.949 & 0.948 & 0.951 & 0.958 & 0.966 & 0.952 & 0.960 & 0.964 \\
CI SE$^*$ & 0.956 & 0.952 & 0.948 & 0.956 & 0.959 & 0.952 & 0.962 & 0.961 & 0.959 & 0.961 & 0.964 & 0.960 \\
SE & 0.258 & 0.202 & 0.160 & 0.336 & 0.263 & 0.208 & 0.254 & 0.197 & 0.156 & 0.329 & 0.256 & 0.201 \\
SE$^*$ & 0.261 & 0.200 & 0.153 & 0.342 & 0.262 & 0.206 & 0.257 & 0.196 & 0.150 & 0.333 & 0.253 & 0.198 \\
\cmidrule(lr){1-13}
\multicolumn{3}{l}{Difference in Means} \\
Bias & 0.154 & 0.179 & 0.216 & 0.168 & 0.192 & 0.233 & 0.001 & 0.001 & 0.005 & 0.004 & 0.001 & 0.004 \\
CI SE & 0.906 & 0.852 & 0.689 & 0.904 & 0.872 & 0.770 & 0.953 & 0.961 & 0.970 & 0.951 & 0.960 & 0.962 \\
CI SE$^*$ & 0.920 & 0.853 & 0.654 & 0.935 & 0.894 & 0.784 & 0.960 & 0.963 & 0.963 & 0.963 & 0.963 & 0.960 \\
SE & 0.244 & 0.186 & 0.144 & 0.319 & 0.246 & 0.191 & 0.240 & 0.182 & 0.140 & 0.316 & 0.242 & 0.188 \\
SE$^*$ & 0.250 & 0.185 & 0.137 & 0.327 & 0.246 & 0.191 & 0.245 & 0.181 & 0.134 & 0.321 & 0.240 & 0.186 \\
\cmidrule(lr){1-13}
\% Excl & 5.035 & 7.816 & 11.25 & 5.035 & 7.816 & 11.25 & 5.035 & 7.816 & 11.25 & 5.035 & 7.816 & 11.25 \\
$r_n$ & 1.395 & 1.518 & 1.623 & 1.395 & 1.518 & 1.623 & 1.395 & 1.518 & 1.623 & 1.395 & 1.518 & 1.623 \\
$\hat\theta$ & 1.518 & 1.770 & 2.089 & 3.406 & 3.723 & 4.119 & 1.432 & 1.650 & 1.927 & 3.288 & 3.563 & 3.907 \\
$k$ & 63 & 100 & 159 & 63 & 100 & 159 & 63 & 100 & 159 & 63 & 100 & 159 \\
\bottomrule
\end{tabular}}
\label{tmain}
\end{table}

\begin{table}[ht]
\centering
\caption{Robustness results}
\resizebox{\columnwidth}{!}{%
\begin{tabular}{lrrrrrrrrrrrr}
\toprule
 & \multicolumn{6}{c}{$c=0.8$} & \multicolumn{6}{c}{$c=1.2$} \\
 \cmidrule(lr){2-7} \cmidrule(lr){8-13}
 & \multicolumn{3}{c}{Indirect Effect} & \multicolumn{3}{c}{Overall Effect} & \multicolumn{3}{c}{Indirect Effect} & \multicolumn{3}{c}{Overall Effect} \\
 \cmidrule(lr){2-4} \cmidrule(lr){5-7} \cmidrule(lr){8-10} \cmidrule(lr){11-13}
$n$ & 500 & 1000 & 2000 & 500 & 1000 & 2000 & 500 & 1000 & 2000 & 500 & 1000 & 2000 \\
\midrule
Bias & 0.092 & 0.085 & 0.083 & 0.102 & 0.092 & 0.090 & 0.044 & 0.035 & 0.033 & 0.050 & 0.038 & 0.036 \\
CI & 0.932 & 0.935 & 0.936 & 0.933 & 0.940 & 0.932 & 0.946 & 0.951 & 0.957 & 0.940 & 0.952 & 0.953 \\
CI SE$^*$ & 0.944 & 0.935 & 0.929 & 0.952 & 0.950 & 0.943 & 0.960 & 0.960 & 0.954 & 0.959 & 0.961 & 0.957 \\
SE & 0.252 & 0.195 & 0.152 & 0.330 & 0.256 & 0.201 & 0.267 & 0.213 & 0.171 & 0.346 & 0.274 & 0.219 \\
SE$^*$ & 0.256 & 0.194 & 0.146 & 0.335 & 0.255 & 0.200 & 0.272 & 0.212 & 0.164 & 0.353 & 0.274 & 0.218 \\
\% Excl & 2.300 & 3.800 & 5.700 & 2.300 & 3.800 & 5.700 & 9.300 & 13.600 & 18.500 & 9.300 & 13.600 & 18.500 \\
$r_n$ & 1.116 & 1.215 & 1.299 & 1.116 & 1.215 & 1.299 & 1.674 & 1.822 & 1.948 & 1.674 & 1.822 & 1.948 \\
$\hat\theta_Q$ & 1.490 & 1.740 & 2.057 & 3.376 & 3.690 & 4.085 & 1.538 & 1.790 & 2.108 & 3.427 & 3.744 & 4.139 \\
$k$ & 63 & 100 & 159 & 63 & 100 & 159 & 63 & 100 & 159 & 63 & 100 & 159 \\
\bottomrule
\end{tabular}}
\label{trobust}
\end{table}

We present results for the indirect and total effects using 5000 simulation draws where within each draw, we redraw potential outcomes and recompute the design-based estimand. In \autoref{tmain}, the ``Spatial Interference'' columns correspond to the outcome model described above, whereas the ``Partial Interference'' columns redefine $w_{ij}=0$ if $i,j$ lie in different clusters to eliminate cross-cluster interference. The ``CI'' rows report the coverage of 95-percent CIs using the indicated standard errors. For our estimator, the ``SE'' row corresponds to standard errors obtained from our variance estimator \eqref{hatsigma}, while for difference in means, it corresponds to conventional cluster-robust SEs. The ``SE$^*$'' rows are the true superpopulation standard errors obtained by taking the standard deviation of the estimator across the simulation draws. As such, SE should be consistent for SE$^*$ rather than conservative, while the CIs have asymptotically conservative coverage for the design-based estimand. Finally, row ``\% Excl'' is the percentage of units that are not well surrounded. 

The results are consistent with the theory. As $n$ grows, the biases of our estimators shrink, while coverage tends to or exceeds the nominal level. The bias of difference in means is more than twice that of our estimators under spatial interference, resulting in severe undercoverage even with the true superpopulation SEs. While the SEs are smaller than those of our estimators, this is not by a significant amount, and the efficiency advantage comes at a large cost to bias under spatial interference. 

\autoref{trobust} presents results for our estimator under spatial interference, but we multiply \eqref{kappa} by $c \in \{0.8,1.2\}$. This is to assess robustness and illustrate a bias-variance trade-off in the choice of $r_n$. The $c=0.8$ columns show that this results in about half the proportion of excluded units relative to \autoref{tmain}. As a result, the bias is higher, resulting in undercoverage. The $c=1.2$ columns show that the proportion of excluded units a little less than doubles. The bias is lower, and as a result, the probability of coverage is higher. On the other hand, the variance increases, as can be seen in the standard error columns.

%----------------------------------------------------------------------
\section{Empirical Application}\label{sapp}
%----------------------------------------------------------------------

We apply our estimator to data from the unconditional cash transfer experiment mentioned in \autoref{eegger}. In the experiment, households eligible for the transfer (the treatment) live in homes with thatched roofs, which is a proxy for poverty. Houses are grouped in villages, which are grouped in ``sublocations.'' See Figure A.2 of \cite{egger2022general} for a map of the study area. The CRT randomizes sublocations (the clusters) into treatment with probability $q=0.5$. Within treatment (control) sublocations, villages are randomized into treatment with probability $p_1=2/3$ ($p_0=1/3$). Within treated villages, all eligible households receive cash transfers totaling 1000 USD, which is about 75 percent of average annual household spending. Sublocations contain 7.8 villages on average (SD 3.9). 

\cite{egger2022general} study the effect of the cash transfers on the following household-level outcomes, which can be grouped into three categories: (1) annualized expenditures on consumption (of food and other purchases described in their footnote 31), non-durables, food alone, temptation goods, and durables; (2) asset stocks, housing value, and land value; and (3) annualized household income, transfers, taxes paid, profits, and wage earnings. The authors report model-assisted estimates of causal effects at the household level. Our analysis will be entirely design-based but at the village level. We aggregate household outcomes to the village level by averaging. 

\cite{egger2022general} estimate the effects of the transfers on the population of eligible households using two main specifications described in their \S 3.2. Their ``RF'' (reduced form) specification is an OLS regression of an outcome on village- and sublocation-level treatment indicators and covariates. They report the coefficient on the village-level indicator, which corresponds to a model-assisted estimate of a village-level direct effect $\theta^*(1,0,2/3,2/3)$ among eligible households.

However, the authors note the potential for interference across sublocations (their quote in \autoref{eegger}). For this reason, their preferred specification is the following ``IV'' (instrumental variables) regression. The main regressors are the amount of cash transferred per capita to the household's village and the amount transferred to neighboring villages within a band of $r-2$ to $r$ km from the ego's village for a range of $r$ values. The corresponding instruments are respectively an indicator for the ego's village being treated and the share of eligible households in the band assigned to treatment. They use a BIC criterion to select the maximum range of $r$, which is 2 km. Using these estimates, the authors compute a model-assisted estimate of the total effect $\theta^*(1,0,2/3,0)$. This compares saturation levels of two-thirds and zero, which is nonparametrically unidentified since control villages have a saturation level of one-third. We instead report our estimates of the total effect $\theta^*(1,0,2/3,1/3)$.

\autoref{tapp} reports the results for a subset of the outcomes, and \autoref{tapp1.5} in the supplementary appendix reports the remainder. Both tables choose $r_n$ in $\hat\theta$ according to \eqref{kappa}, which results in $r_n=1.6$ and 39.66 percent of units not well surrounded. In \autoref{sappSA} we describe how we construct cluster radii $R_j$ used in this formula. The $\hat\theta^+$ columns correspond to difference-in-means estimates with clustered standard errors. 

We find that the difference-in-means estimates of the direct and total effects are comparable to the RF and IV estimates, respectively, despite the distinctions outlined above. Our estimators find larger direct and total effects. Decomposing the total effect into direct and indirect effects, we find that the latter are substantially smaller in magnitude with large standard errors. Compared to difference in means, our estimates tend to be larger in magnitude but with larger standard errors due to the restriction to well-surrounded units.

\begin{table}[ht]
\centering
\caption{Effects on eligibles}
\resizebox{\columnwidth}{!}{%
\begin{threeparttable}
\begin{tabular}{lrrrrrrrr}
\toprule
{} & \multicolumn{2}{r}{Direct Effect} & \multicolumn{2}{r}{Indirect Effect} & \multicolumn{2}{r}{Total Effect} & \multicolumn{2}{r}{Egger et al.} \\
{} & {$\hat\theta$} & {$\hat\theta^+$} & {$\hat\theta$} & {$\hat\theta^+$} & {$\hat\theta$} & {$\hat\theta^+$} & {RF} & {IV} \\
\midrule
consumption & 390.70 & 253.93 & 9.51 & 64.43 & 400.21 & 318.37 & 293.59 & 338.57 \\
 & {\footnotesize (109.19)} & {\footnotesize (76.32)} & {\footnotesize (155.56)} & {\footnotesize (103.07)} & {\footnotesize (158.01)} & {\footnotesize (102.18)} & {\footnotesize (60.11)} & {\footnotesize (109.38)} \\
non-durable & 248.71 & 151.99 & 43.18 & 61.08 & 291.90 & 213.07 & 187.65 & 227.2 \\
 & {\footnotesize (104.20)} & {\footnotesize (66.39)} & {\footnotesize (139.26)} & {\footnotesize (96.22)} & {\footnotesize (140.41)} & {\footnotesize (94.55)} & {\footnotesize (58.59)} & {\footnotesize (99.63)} \\
assets & 249.58 & 180.43 & -35.18 & -11.43 & 214.40 & 168.99 & 178.78 & 183.38 \\
 & {\footnotesize (59.26)} & {\footnotesize (39.32)} & {\footnotesize (104.75)} & {\footnotesize (66.00)} & {\footnotesize (98.33)} & {\footnotesize (59.98)} & {\footnotesize (24.66)} & {\footnotesize (44.26)} \\
housing & 422.82 & 376.36 & 18.77 & 31.02 & 441.60 & 407.38 & 376.92 & 477.29 \\
 & {\footnotesize (49.74)} & {\footnotesize (31.12)} & {\footnotesize (93.55)} & {\footnotesize (51.73)} & {\footnotesize (86.81)} & {\footnotesize (48.74)} & {\footnotesize (26.37)} & {\footnotesize (38.8)} \\
income & 132.89 & 95.02 & 64.90 & 21.01 & 197.80 & 116.03 & 79.43 & 135.7 \\
 & {\footnotesize (120.95)} & {\footnotesize (58.97)} & {\footnotesize (159.60)} & {\footnotesize (91.71)} & {\footnotesize (154.70)} & {\footnotesize (87.36)} & {\footnotesize (43.8)} & {\footnotesize (92.1)} \\
earnings & 59.35 & 45.70 & 22.31 & 6.46 & 81.67 & 52.16 & 42.43 & 73.66 \\
 & {\footnotesize (78.90)} & {\footnotesize (37.67)} & {\footnotesize (123.04)} & {\footnotesize (66.82)} & {\footnotesize (119.93)} & {\footnotesize (60.62)} & {\footnotesize (32.23)} & {\footnotesize (60.82)} \\
\bottomrule
\end{tabular}
\begin{tablenotes}[para,flushleft]
  \footnotesize 653 villages (units), 84 sublocations (clusters). Standard errors are in parentheses. Column RF (IV) is the reduced form (IV) estimate of the direct (total) effect from Table I, column 1 (2) of \cite{egger2022general}, $\hat\theta$ is our estimate, and $\hat\theta^+$ is difference in means. Our estimates use $r_n=1.6$, which results in 39.66 percent of units not being well surrounded.
\end{tablenotes}
\end{threeparttable}}
\label{tapp}
\end{table}

\autoref{tapp2} in the supplementary appendix reports results for $r_n=2$, resulting 57.89 percent of units not being well surrounded. This choice of $r_n$ coincides with the largest range of $r$ selected by the BIC procedure of \cite{egger2022general}. Our estimates and standard errors become larger still in magnitude relative to difference in means, but the results are qualitatively similar.

To estimate spillover effects, \cite{egger2022general} rerun their IV specification using only non-eligible households, which did not receive any transfers (their \S 3.3). In \autoref{tappn} of the supplementary appendix, we compare their results with design-based estimates of the overall effect $\theta^*(\emptyset, \emptyset, 2/3, 1/3)$ on the population of non-eligibles. The effect sizes of our estimators and those of difference in means are fairly similar in magnitude to their IV results, but the standard errors are large. Combined with the results in \autoref{tapp}, we ultimately find strong direct effects of the cash transfers but weaker evidence for spillover effects compared to \cite{egger2022general}.

%----------------------------------------------------------------------
\section{Conclusion}\label{sconclude}
%----------------------------------------------------------------------

When interference occurs across clusters, conventional analyses of CRTs suffer from bias induced by units near cluster boundaries. To reduce bias at the analysis stage, we provide in \autoref{sest} an estimator $\hat\theta$ that improves upon the fried-egg design by excluding from estimation units that are not surrounded by clusters assigned to the same treatment arm. This may be employed as a robustness check for difference in means. To reduce bias at the design stage, we propose a rate-optimal formula for the number of clusters $k$ in \eqref{optk}. Unlike the standard power analysis that assumes partial interference, our choice balances power against the need to reduce bias and accounts for the density of units in the spatial region. Given $k$, we suggest automating cluster construction using $k$-medoids and prove that the resulting clusters are balanced and globular, thereby approximately minimizing the number of units near boundaries. We also provide valid design-based standard errors.

Under the conventional superpopulation, partial interference framework, power calculations for choosing $k$ require prior knowledge of the intracluster correlation coefficient (ICC). Under our design-based framework, the optimal choice of $k$ requires knowledge of $\gamma$, the speed at which interference decays with distance, rather than the ICC. Absent domain knowledge, one can conservatively set $\gamma$ to the dimension of the spatial region. We discuss in \autoref{sgamma} how to obtain less conservative estimates via modeling or prior data. 

%----------------------------------------------------------------------

\newpage
\part{Supplementary Appendix}

\makeatletter
\@addtoreset{section}{part}
\makeatother
\renewcommand{\thesection}{SA.\arabic{section}} % custom SA section numbering
\setcounter{section}{0}
\numberwithin{equation}{section} % include section number in equation numbering
\numberwithin{table}{section}

%----------------------------------------------------------------------
\section{\texorpdfstring{$k$}{k}-Medoids}\label{skmed}
%----------------------------------------------------------------------

The $k$-medoids algorithm selects a set of $k$ units in $\N_n$ -- the {\em medoids} or cluster centroids -- to minimize the total distance between all units and their nearest medoids. This creates medoids that are spatially well separated, as shown in \autoref{lsep}. Units are then grouped into clusters based on their closest medoids. 

Whereas $k$-means allows centroids to be any element of $\X$, $k$-medoids constrains the centroids to the data $\N_n \subseteq \X$, but in practice the algorithms tend to produce similar output. Our theoretical results pertain to the implementation in \autoref{amedoids}, the ``partitioning around medoids'' algorithm. This is simple to understand, though it has complexity $O(k(n-k)^2)$ compared to an $O(n^2)$ runtime under the fastest known implementation \citep{schubert2021fast}.

For any set of candidate medoids $\mathcal{M} \subseteq \N_n$, let $\text{cost}(\mathcal{M}) = \sum_{i \in \N_n} \min_{m \in \mathcal{M}} \rho(i,m)$, the total distance between units and their nearest medoids.

\begin{secalgo}[$k$-Medoids]\label{amedoids} \hfill
\begin{enumerate}
  \item Initialize an arbitrary set of $k$ medoids $\mathcal{M} = \{m_j\}_{j=1}^k \subseteq \N_n$. Given any $\mathcal{M}$, the associated set of clusters is $\C_n = \{C_j\}_{j=1}^k$, where $C_j = \{i \in \N_n\colon m_j = \argmin_{m \in \mathcal{M}} \rho(i,m)\}$.
  \item While there exist $m \in \mathcal{M}$ and $o \in \N_n\backslash\mathcal{M}$ such that $\text{cost}(\mathcal{M} \backslash \{m\} \cup \{o\}) < \text{cost}(\mathcal{M})$, replace $m$ with the $o$ that minimizes $\text{cost}(\mathcal{M} \backslash \{m\} \cup \{o\})$.
  \item Output $\C_n$.
\end{enumerate}
\end{secalgo}

\noindent In other words, this chooses the set of medoids $\mathcal{M}$ that minimizes cost by iteratively swapping out a candidate medoid with a better unit that reduces cost.

%----------------------------------------------------------------------
\section{Extensions}\label{sext}
%----------------------------------------------------------------------

%---------------------------------------
\subsection{Heterogeneous Density}\label{shetdens}
%---------------------------------------

As discussed in \autoref{sinter}, \autoref{aspace}(a) restricts the extent to which density can vary across the study region. To accommodate larger variation in density, such as when the study region encompasses both urban and rural areas, we suggest first partitioning the region $\N_n$ into subregions that are relatively homogeneous in density. This may be done manually, for instance using cartographic and demographic information. It can also be done using a density-based clustering algorithm; see \cite{bhattacharjee2021survey} and especially \cite{kriegel2011density} for surveys of this literature. 

Call the density-homogeneous subregions $\mathcal{S}_1, \ldots, \mathcal{S}_m$. For each $\mathcal{S}_j$ the researcher can compute the optimal number of clusters $k_j$ given in \eqref{optk}, subdivide $\mathcal{S}_j$ into $k_j$ clusters $C_{j1}, \ldots, C_{jk_j}$, say using $k$-medoids, and cluster-randomize across the collection of all clusters $\{C_{j\ell}\colon \ell=1,\ldots k_j, j=1,\ldots, m\}$. Lastly, we recommend modifying the formula for $r_n$ since cluster radii may vary substantially across subregions with differing densities. In the definition of $T_{ti}$, replace $r_n$ with half the median cluster radius among clusters in the subregion containing $i$.

%---------------------------------------
\subsection{Determining \texorpdfstring{$\gamma$}{gamma}}\label{sgammax}
%---------------------------------------

Estimating the following spatial moving average model with prior data can be a starting point for determining $\gamma$:
\begin{equation*}
  Y_i = \sum_{j\in\mathcal{N}_n} \rho(i,j)^{-\eta}(\alpha + D_j\beta_j + \varepsilon_j).
\end{equation*}

\noindent This satisfies \autoref{aani} with $\gamma = \eta-2$ \citep[][Proposition 1]{leung2022rate}, so $\gamma$ can be backed out from an estimate of $\eta$.

It would also be useful to develop designs for nonparametrically estimating $\gamma$. A preliminary proposal is the following. Given $k$ clusters, say constructed using $k$-medoids, denote by $m_j$ the centroid of a cluster $C_j$. Randomize clusters to the following $T$ treatment arms. In the 0th arm, we assign all units to control. In the $t$th treatment arm for $t\geq 1$, we only assign units $i$ for which $\rho(i,m) \in (t, t+1]$ to treatment, where $m$ is the centroid of the cluster in question. That is, we treat only units in a ring at a certain distance from the centroid. Let $\hat\theta_t$ denote the difference-in-means estimate comparing average outcomes of units near centroids of clusters assigned to arm $t\geq 1$ with those assigned to arm 0. Then under \autoref{aani}, $\hat\theta_t$ should decay like $t^{-\gamma}$, producing a sort of ``causal covariogram.'' We may then regress $\log \hat\theta_t$ on $\log t$ to estimate $-\gamma$.

%---------------------------------------
\subsection{Bias-Aware CIs}\label{sBA}
%---------------------------------------

Recall from \autoref{sasymp} that the ``undersmoothed'' design chooses $k_n$ smaller than the optimal rate, which is why \eqref{optk} uses a strict lower bound $\tilde\gamma$ in place of $\gamma$. By analogy to nonparametric regression, an alternative is to choose $k$ rate-optimally using \eqref{optk} with $\gamma$ in place of $\tilde\gamma$ and use the following ``bias-aware'' confidence interval (CI) in place of \eqref{undersmoothCI}:
\begin{equation}
  \hat\theta \pm \big( 4c\sqrt{k}r_n^{-\gamma} + 1.96 \cdot \hat\sigma k^{-1/2} \big), \label{biasawareCI}
\end{equation}

\noindent where $c$ is given in \autoref{aani}. \cite{faridani2023rate} propose a similar bias-aware CI for the GATE. Compared to \eqref{undersmoothCI}, \eqref{biasawareCI} should have better finite-sample coverage, although its implementation requires prior knowledge of both $\gamma$ and $c$. 

In addition to knowledge of these parameters, suppose the researcher has preliminary consistent estimates of $\hat\theta$ and $\hat\sigma^2$, say from a pilot study (which is more plausible in a superpopulation setup). Then following \cite{armstrong2018optimal}, they can pick $k$ to minimize the length of the bias-aware CI \eqref{biasawareCI}. Note that $r_n$ is an implicit function of $k$, which can be traced out using grid search.

The idea behind \eqref{biasawareCI} is that, by \autoref{tclt},
\begin{equation*}
  \sigma_n^{-1} \sqrt{k_n} \big(\hat\theta - \theta^*\big) = \underbrace{\sigma_n^{-1} \sqrt{k_n} \big(\hat\theta - \bar{\theta}\big)}_{\dlimarrow \N(0,1)} + \sigma_n^{-1} \underbrace{\sqrt{k_n}(\bar{\theta} - \theta^*)}_{\mathcal{B}_n^*}.
\end{equation*}

\noindent As shown in the proof of \autoref{tQ}, specifically the argument preceding \eqref{fbi},
\begin{equation*}
  \abs{\mathcal{B}_n^*} \leq 4c\sqrt{k_n}r_n^{-\gamma}.
\end{equation*}

\noindent This provides a worst-case bound on the bias that \eqref{biasawareCI} incorporates.

%----------------------------------------------------------------------
\section{Variance Estimators}\label{svar}
%----------------------------------------------------------------------

We next discuss how $\hat\sigma^2$ relates to variance estimators in the literature. Define $\bm{A}(1)$ ($\bm{A}(2)$) as the symmetric, $n\times n$ matrix with $ij$th entry $A_{ij}(1)$ ($A_{ij}(2)$), and recall that $\hat\sigma^2(2)$ is the cluster-robust variance estimator, while $\hat\sigma^2(1)$ is analogous to the \cite{leung2022rate} variance estimator. The proof of \autoref{tvar} shows that the cluster-robust variance estimator can be asymptotically decomposed as $\hat\sigma^2(2) = \sigma_n^2 + \mathcal{B}_n + o_p(1)$ where
\begin{equation}
  \mathcal{B}_n = \frac{k_n}{n^2} \bm{\mu}' \bm{A}(2) \bm{\mu} \label{Bn}
\end{equation}

\noindent and $\bm{\mu}$ is the $n$-dimensional vector with $i$th component $(\mu_{1i}-\mu_{0i}) - (\mu_1 - \mu_0)$ for $\mu_{ti} = \E[Y_i \mid T_{ti}=1]$. Because $\bm{A}(2)$ is block-diagonal, it is positive semidefinite, so $\mathcal{B}_n \geq 0$ for any $n$. Hence, $\hat\sigma^2(2)$ is asymptotically conservative.

The proof further shows that $\hat\sigma^2(1) = \sigma_n^2 + \mathcal{B}_n(1) + o_p(1)$ where $\mathcal{B}_n(1)$ is obtained by replacing $\bm{A}(2)$ with $\bm{A}(1)$ in \eqref{Bn}. This replacement adds nonzero off-diagonal elements to $\bm{A}(2)$, so $\bm{A}(1)$ is not block-diagonal and hence not guaranteed to be positive semidefinite. Accordingly, Leung imposes additional conditions to show that $\mathcal{B}_n(1) \plimarrow c \geq 0$ in the superpopulation. We find, however, that $\mathcal{B}_n(1) = \mathcal{B}_n + o_p(1)$ without any conditions on the superpopulation, so in fact $\hat\sigma^2(1)$ is asymptotically conservative in a purely design-based setup. \cite{faridani2023rate} provide a different approach to conservative inference for the GATE based on bounding $\mathcal{B}_n$.

The positive-semidefiniteness of the ``kernel'' $\bm{A}(2)$ results in both a non-negative variance estimator for any $n$ and asymptotic conservativeness. This mirrors the corresponding insight for HAC variance estimators, that positive semidefinite kernels ensure conservativeness in finite-population models. This fact was first pointed out by \cite{leung2019causal} and has since been exploited by other papers. (In his case of network-dependent data, the difficulty with using positive-semidefinite HAC kernels is that they are sloped and tend to severely over-reject in finite samples, which is why \cite{leung2022causal} recommends use of the uniform kernel even though it is not positive semidefinite. For spatial data, a variety of positive semidefinite HAC kernels exist, and these tend to have less severe issues with over-rejection relative to the network case.)

\cite{hudgens2008toward} propose variance estimators for difference-in-means type estimators under partial interference. Their theory relies on an additional stratified interference assumption, which says that potential outcomes only depend on the ego's treatment assignment and the proportion of treated units in the ego's cluster, in which case $Y_i = Y_i(D_i, p)$ (since they assume complete randomization). Because $p$ is a constant, the dependence structure is the same as under no interference since outcomes within a cluster are only correlated if treatment assignments are. The variance estimators proposed by \cite{hudgens2008toward} heavily rely on this structure. In contrast, we do not impose partial, let alone stratified, interference, so our setting features both within- and cross-cluster dependence in outcomes. It is therefore critical to account for additional covariance terms that are absent in the \cite{hudgens2008toward} variance formula to avoid anti-conservativeness. Our estimator does so through the terms involving $A_{ij}(u)$ for $i\neq j$.

In the econometric literature, the standard approach under partial interference is clustering standard errors \citep{baird2018optimal}. These allow for arbitrary within-cluster dependence and are valid even in the absence of stratified interference. A new finding of our paper is that clustered standard errors are also valid in a design-based setting with cross-cluster interference. However, we suggest combining them with $\hat\sigma^2(1)$ to better capture second-order covariance terms, as discussed in \autoref{sasymp}.

%----------------------------------------------------------------------
\section{Empirical Application}\label{sappSA}
%----------------------------------------------------------------------

We define $R_j$ from \eqref{kappa} as half the largest distance between households in any pair of villages within cluster $j$. To compute this, we utilize supplemental data provided by Dennis Egger on distances between village centroids and publicly available data from \cite{egger2022general} to estimate village radii. We first compute the largest distance between the centroids of any pair of villages $a,b$ within a given cluster, which we denote by $\Delta_{a,b}$. Since this does not account for village size, we estimate for each village $a$ its radius $\delta_a$, which is the furthest distance between a household in the village and its centroid. This data is publicly available in 1 km increments. For each cluster $C_j$, we define $R_j = \max_{a,b \in C_j} (\Delta_{a,b} + \delta_a + \delta_b)/2$. The average radius across clusters is 3.12 with a standard deviation of 0.81. 

\begin{table}[ht]
\centering
\caption{Effects on eligibles, $r_n=1.6$}
\resizebox{\columnwidth}{!}{%
\begin{threeparttable}
\begin{tabular}{lrrrrrrrr}
\toprule
{} & \multicolumn{2}{r}{Direct Effect} & \multicolumn{2}{r}{Indirect Effect} & \multicolumn{2}{r}{Total Effect} & \multicolumn{2}{r}{Egger et al.} \\
{} & {$\hat\theta$} & {$\hat\theta^+$} & {$\hat\theta$} & {$\hat\theta^+$} & {$\hat\theta$} & {$\hat\theta^+$} & {RF} & {IV} \\
\midrule
food & 110.02 & 67.33 & 34.11 & 60.08 & 144.14 & 127.42 & 72.04 & 133.84 \\
 & {\footnotesize (63.00)} & {\footnotesize (47.31)} & {\footnotesize (81.92)} & {\footnotesize (62.21)} & {\footnotesize (84.04)} & {\footnotesize (58.73)} & {\footnotesize (36.96)} & {\footnotesize (63.99)} \\
temptation & -5.27 & -0.21 & 12.17 & 5.07 & 6.90 & 4.85 & 6.55 & 5.91 \\
 & {\footnotesize (16.32)} & {\footnotesize (9.39)} & {\footnotesize (19.07)} & {\footnotesize (10.79)} & {\footnotesize (14.79)} & {\footnotesize (9.95)} & {\footnotesize (5.79)} & {\footnotesize (8.82)} \\
durable & 123.08 & 87.82 & -23.94 & 9.53 & 99.14 & 97.35 & 95.09 & 109.01 \\
 & {\footnotesize (27.92)} & {\footnotesize (23.78)} & {\footnotesize (42.61)} & {\footnotesize (17.99)} & {\footnotesize (46.97)} & {\footnotesize (20.43)} & {\footnotesize (12.64)} & {\footnotesize (20.24)} \\
land & -47.09 & -32.65 & 307.03 & 199.58 & 259.93 & 166.93 & 51.28 & 158.47 \\
 & {\footnotesize (422.26)} & {\footnotesize (271.27)} & {\footnotesize (608.38)} & {\footnotesize (376.02)} & {\footnotesize (538.61)} & {\footnotesize (300.98)} & {\footnotesize (186.22)} & {\footnotesize (260.91)} \\
transfers & 11.68 & 0.06 & -11.56 & -1.54 & 0.11 & -1.47 & -1.68 & -7.43 \\
 & {\footnotesize (11.86)} & {\footnotesize (7.20)} & {\footnotesize (22.84)} & {\footnotesize (11.76)} & {\footnotesize (21.94)} & {\footnotesize (10.85)} & {\footnotesize (6.81)} & {\footnotesize (13.06)} \\
tax & -1.52 & 0.97 & 1.27 & 2.02 & -0.25 & 3.00 & 1.94 & -0.09 \\
 & {\footnotesize (1.91)} & {\footnotesize (1.36)} & {\footnotesize (3.92)} & {\footnotesize (2.32)} & {\footnotesize (3.71)} & {\footnotesize (2.22)} & {\footnotesize (1.28)} & {\footnotesize (2.02)} \\
profits & 48.70 & 41.40 & 31.90 & 19.85 & 80.60 & 61.25 & 26.24 & 35.85 \\
 & {\footnotesize (50.16)} & {\footnotesize (34.97)} & {\footnotesize (73.30)} & {\footnotesize (45.79)} & {\footnotesize (68.68)} & {\footnotesize (44.39)} & {\footnotesize (23.67)} & {\footnotesize (47.66)} \\
\bottomrule
\end{tabular}
\begin{tablenotes}[para,flushleft]
  \footnotesize 653 villages (units), 84 sublocations (clusters). Standard errors are in parentheses. Column RF (IV) is the reduced form (IV) estimate of the overall effect from Table I, column 1 (2) of \cite{egger2022general}, $\hat\theta$ is our estimate, and $\hat\theta^+$ is difference in means. Our estimates use $r_n=1.6$, which results in 39.66 percent of units not being well surrounded.
\end{tablenotes}
\end{threeparttable}}
\label{tapp1.5}
\end{table}

\begin{table}[ht]
\centering
\caption{Effects on eligibles, $r_n=2$}
\resizebox{\columnwidth}{!}{%
\begin{threeparttable}
\begin{tabular}{lrrrrrrrr}
\toprule
{} & \multicolumn{2}{r}{Direct Effect} & \multicolumn{2}{r}{Indirect Effect} & \multicolumn{2}{r}{Total Effect} & \multicolumn{2}{r}{Egger et al.} \\
{} & {$\hat\theta$} & {$\hat\theta^+$} & {$\hat\theta$} & {$\hat\theta^+$} & {$\hat\theta$} & {$\hat\theta^+$} & {RF} & {IV} \\
\midrule
consumption & 458.57 & 253.93 & 63.41 & 64.43 & 521.99 & 318.37 & 293.59 & 338.57 \\
 & {\footnotesize (112.05)} & {\footnotesize (76.32)} & {\footnotesize (185.15)} & {\footnotesize (103.07)} & {\footnotesize (185.76)} & {\footnotesize (102.18)} & {\footnotesize (60.11)} & {\footnotesize (109.38)} \\
non-durable & 302.79 & 151.99 & 92.02 & 61.08 & 394.81 & 213.07 & 187.65 & 227.2 \\
 & {\footnotesize (110.70)} & {\footnotesize (66.39)} & {\footnotesize (171.83)} & {\footnotesize (96.22)} & {\footnotesize (167.98)} & {\footnotesize (94.55)} & {\footnotesize (58.59)} & {\footnotesize (99.63)} \\
food & 141.56 & 67.33 & 69.97 & 60.08 & 211.54 & 127.42 & 72.04 & 133.84 \\
 & {\footnotesize (70.48)} & {\footnotesize (47.31)} & {\footnotesize (107.03)} & {\footnotesize (62.21)} & {\footnotesize (103.43)} & {\footnotesize (58.73)} & {\footnotesize (36.96)} & {\footnotesize (63.99)} \\
temptation & -1.27 & -0.21 & 7.51 & 5.07 & 6.23 & 4.85 & 6.55 & 5.91 \\
 & {\footnotesize (21.22)} & {\footnotesize (9.39)} & {\footnotesize (20.02)} & {\footnotesize (10.79)} & {\footnotesize (16.87)} & {\footnotesize (9.95)} & {\footnotesize (5.79)} & {\footnotesize (8.82)} \\
durable & 135.67 & 87.82 & -21.12 & 9.53 & 114.55 & 97.35 & 95.09 & 109.01 \\
 & {\footnotesize (28.50)} & {\footnotesize (23.78)} & {\footnotesize (45.92)} & {\footnotesize (17.99)} & {\footnotesize (50.43)} & {\footnotesize (20.43)} & {\footnotesize (12.64)} & {\footnotesize (20.24)} \\
assets & 251.57 & 180.43 & -1.78 & -11.43 & 249.79 & 168.99 & 178.78 & 183.38 \\
 & {\footnotesize (62.15)} & {\footnotesize (39.32)} & {\footnotesize (121.35)} & {\footnotesize (66.00)} & {\footnotesize (112.25)} & {\footnotesize (59.98)} & {\footnotesize (24.66)} & {\footnotesize (44.26)} \\
housing & 425.28 & 376.36 & -19.37 & 31.02 & 405.90 & 407.38 & 376.92 & 477.29 \\
 & {\footnotesize (48.04)} & {\footnotesize (31.12)} & {\footnotesize (101.37)} & {\footnotesize (51.73)} & {\footnotesize (98.87)} & {\footnotesize (48.74)} & {\footnotesize (26.37)} & {\footnotesize (38.8)} \\
land & -381.97 & -32.65 & 667.52 & 199.58 & 285.54 & 166.93 & 51.28 & 158.47 \\
 & {\footnotesize (533.21)} & {\footnotesize (271.27)} & {\footnotesize (682.49)} & {\footnotesize (376.02)} & {\footnotesize (594.10)} & {\footnotesize (300.98)} & {\footnotesize (186.22)} & {\footnotesize (260.91)} \\
income & 195.64 & 95.02 & 142.34 & 21.01 & 337.98 & 116.03 & 79.43 & 135.7 \\
 & {\footnotesize (133.25)} & {\footnotesize (58.97)} & {\footnotesize (170.94)} & {\footnotesize (91.71)} & {\footnotesize (171.01)} & {\footnotesize (87.36)} & {\footnotesize (43.8)} & {\footnotesize (92.1)} \\
transfers & 3.87 & 0.06 & -12.72 & -1.54 & -8.85 & -1.47 & -1.68 & -7.43 \\
 & {\footnotesize (14.56)} & {\footnotesize (7.20)} & {\footnotesize (24.96)} & {\footnotesize (11.76)} & {\footnotesize (23.93)} & {\footnotesize (10.85)} & {\footnotesize (6.81)} & {\footnotesize (13.06)} \\
tax & 0.11 & 0.97 & 0.09 & 2.02 & 0.20 & 3.00 & 1.94 & -0.09 \\
 & {\footnotesize (1.95)} & {\footnotesize (1.36)} & {\footnotesize (4.30)} & {\footnotesize (2.32)} & {\footnotesize (4.63)} & {\footnotesize (2.22)} & {\footnotesize (1.28)} & {\footnotesize (2.02)} \\
profits & 32.37 & 41.40 & 32.10 & 19.85 & 64.48 & 61.25 & 26.24 & 35.85 \\
 & {\footnotesize (60.65)} & {\footnotesize (34.97)} & {\footnotesize (73.81)} & {\footnotesize (45.79)} & {\footnotesize (76.64)} & {\footnotesize (44.39)} & {\footnotesize (23.67)} & {\footnotesize (47.66)} \\
earnings & 134.09 & 45.70 & 98.88 & 6.46 & 232.97 & 52.16 & 42.43 & 73.66 \\
 & {\footnotesize (99.94)} & {\footnotesize (37.67)} & {\footnotesize (122.24)} & {\footnotesize (66.82)} & {\footnotesize (123.66)} & {\footnotesize (60.62)} & {\footnotesize (32.23)} & {\footnotesize (60.82)} \\
\bottomrule
\end{tabular}
\begin{tablenotes}[para,flushleft]
  \footnotesize 653 villages (units), 84 sublocations (clusters). Standard errors are in parentheses. Column RF (IV) is the reduced form (IV) estimate of the overall effect from Table I, column 1 (2) of \cite{egger2022general}, $\hat\theta$ is our estimate, and $\hat\theta^+$ is difference in means. Our estimates use $r_n=2$, which results in 57.89 percent of units not being well surrounded.
\end{tablenotes}
\end{threeparttable}}
\label{tapp2}
\end{table}

\begin{table}[ht]
\centering
\caption{Overall effect on noneligibles}
\begin{threeparttable}
\begin{tabular}{rrrrr}
\toprule
 & $\hat\theta$, $r_n{=}2$ & $\hat\theta$, $r_n{=}1.6$ & $\hat\theta^+$ & IV \\
\midrule
consumption & 247.68 & 228.51 & 211.08 & 334.77 \\
 & {\footnotesize (145.79)} & {\footnotesize (139.91)} & {\footnotesize (92.16)} & {\footnotesize (123.2)} \\
non-durable & 209.57 & 195.41 & 195.38 & 317.62 \\
 & {\footnotesize (140.05)} & {\footnotesize (135.67)} & {\footnotesize (87.14)} & {\footnotesize (119.76)} \\
food & 78.62 & 94.80 & 94.25 & 133.3 \\
 & {\footnotesize (81.01)} & {\footnotesize (76.56)} & {\footnotesize (48.29)} & {\footnotesize (58.56)} \\
temptation & 6.75 & 2.30 & 6.57 & -0.68 \\
 & {\footnotesize (7.35)} & {\footnotesize (7.03)} & {\footnotesize (5.24)} & {\footnotesize (6.5)} \\
durable & 30.56 & 22.79 & 8.51 & 8.44 \\
 & {\footnotesize (16.01)} & {\footnotesize (10.56)} & {\footnotesize (6.76)} & {\footnotesize (12.5)} \\
assets & 159.35 & 80.62 & 41.85 & 133.06 \\
 & {\footnotesize (127.86)} & {\footnotesize (123.09)} & {\footnotesize (68.68)} & {\footnotesize (78.33)} \\
housing & 111.72 & 198.86 & 340.65 & 80.65 \\
 & {\footnotesize (510.87)} & {\footnotesize (338.18)} & {\footnotesize (208.72)} & {\footnotesize (215.81)} \\
land & 1132.78 & 353.26 & 133.88 & 544.85 \\
 & {\footnotesize (1013.78)} & {\footnotesize (779.26)} & {\footnotesize (407.10)} & {\footnotesize (459.57)} \\
income & 250.09 & 167.14 & 143.48 & 224.96 \\
 & {\footnotesize (142.24)} & {\footnotesize (140.37)} & {\footnotesize (78.78)} & {\footnotesize (85.98)} \\
transfers & 12.45 & 9.48 & 13.73 & 8.85 \\
 & {\footnotesize (19.43)} & {\footnotesize (17.13)} & {\footnotesize (9.58)} & {\footnotesize (19.11)} \\
tax & -0.65 & 3.11 & 2.60 & 1.68 \\
 & {\footnotesize (3.92)} & {\footnotesize (2.72)} & {\footnotesize (1.63)} & {\footnotesize (2.02)} \\
profits & 80.27 & 92.67 & 88.17 & 36.37 \\
 & {\footnotesize (67.82)} & {\footnotesize (64.10)} & {\footnotesize (34.44)} & {\footnotesize (44.88)} \\
earnings & 145.84 & 43.22 & 33.88 & 182.63 \\
 & {\footnotesize (101.11)} & {\footnotesize (99.00)} & {\footnotesize (60.40)} & {\footnotesize (65.53)} \\
\bottomrule
\end{tabular}
\begin{tablenotes}[para,flushleft]
  \footnotesize $653$ villages (units), $84$ sublocations (clusters). Standard errors are in parentheses. Column IV is the IV estimate of the overall effect from Table I, column 3 of \cite{egger2022general}.
\end{tablenotes}
\end{threeparttable}
\label{tappn}
\end{table}

%----------------------------------------------------------------------
\section{Auxiliary Lemmas}
%----------------------------------------------------------------------

The first two lemmas establish properties of $k$-medoids clusters for $k=k_n$ possibly diverging. 

\begin{seclemma}[Separation]\label{lsep}
  Suppose clusters are generated by \autoref{amedoids}. Under \autoref{aspace}(a), there exists a positive sequence $\{\ell_n\}_{n\in\mathbb{N}}$ bounded away from zero such that for all $n$ sufficiently large and medoids $i_n,j_n$ generated by $k_n$-medoids, $\rho(i_n,j_n) \geq \ell_n (n/(k_n\xi_n))^{1/d}$.
\end{seclemma}

\begin{seclemma}[Radius]\label{lradius}
  Suppose clusters are generated by \autoref{amedoids}. Consider any sequence $\{C_n\}_{n\in\mathbb{N}}$ with $C_n \in \C_n$ for each $n$. Let $i_n$ be the medoid of $C_n$ and $R_n^* = \max\{\rho(i_n,j)\colon j \in C_n\}$. Under \autoref{aspace}(a), $R_n^* \sim (n/(k_n\xi_n))^{1/d}$, so $\abs{\N(i_n, R_n^*)} \precsim n/k_n$. 
\end{seclemma}

\begin{seclemma}\label{lintersect}
  Under Assumptions \ref{aspace}(a) and \ref{aclus}, $\max_{i\in\N_n} \phi_i \precsim 1$ for $\phi_i$ defined in \autoref{roverlap}. Consequently, under Assumptions \ref{aCRT} and \ref{aoverlap}, $\min_i p_{ti} \succsim 1$ for any $t\in\{0,1\}$.
\end{seclemma} 

\begin{seclemma}\label{lnbhd}
  Under Assumptions \ref{aspace}(a) and \ref{aclus}, $\max_{i\in\N_n} \abs{\Lambda_i} \precsim n/k_n$ for $\Lambda_i$ defined in \eqref{Lambdai}.
\end{seclemma}

\begin{seclemma}\label{lkappa}
  Let $\hat\kappa_t = n^{-1} \sum_{i\in\N_n} T_{ti}/p_{ti}$. Under Assumptions \ref{aspace}(a) and \ref{aCRT}--\ref{aoverlap}, $\hat\kappa_t - 1 \precsim k_n^{-1/2}$ for any $t\in\{0,1\}$.
\end{seclemma}

Let $T_{ti}^+$ equal $T_{ti}$ with $r_n$ set to zero, $\hat{p}_t^+ = n^{-1} \sum_{i\in\N_n} T_{ti}^+$, and $p_t^+ = \E[T_{ti}^+]$, which does not depend on $i$ under \autoref{aCRT}.

\begin{seclemma}\label{lpt}
  Under Assumptions \ref{aspace}(a) and \ref{aCRT}--\ref{aoverlap}, $\hat{p}_t^+ - p_t^+ = O_p(k_n^{-1/2})$ and $p_t^+/\hat{p}_t^+ \plimarrow 1$ for any $t\in\{0,1\}$.
\end{seclemma}

%----------------------------------------------------------------------
\section{Proofs}\label{sproofs}
%----------------------------------------------------------------------

%---------------------------------------
\subsection{\autoref{tkmed}}
%---------------------------------------

Let each cluster $C_n$'s ``centroid'' $i_n$ be its medoid and $U_n \equiv \max_j R_j$ for $R_j$ defined prior to \eqref{kappa}, so $C_n \subseteq \N(i_n, U_n)$. By \autoref{lradius}, $U_n \precsim (n/(k_n\xi_n))^{1/d}$. By \autoref{lsep}, there exists $\Delta_n \succsim (n/(k_n\xi_n))^{1/d}$ such that any pair of medoids is physically separated by at least $\Delta_n$. Then by definition of $k$-medoid clusters, for $L_n = \Delta_n/2$, $\N(i_n,L_n) \subseteq C_n$. \qed

%---------------------------------------
\subsection{\autoref{tQ}}\label{stQ}
%---------------------------------------

%---------------------
\subsubsection{Reduction to HT}
%---------------------

Define the following Horvitz-Thompson analog of $\hat\theta$:
\begin{equation}
  \tilde\theta = \frac{1}{n} \sum_{i\in\N_n} \tilde Z_i \quad\text{for}\quad \tilde Z_i = \left( \frac{T_{1i}}{p_{1i}} - \frac{T_{0i}}{p_{0i}} \right) Y_i. \label{HT}
\end{equation}

\noindent This is unbiased for $\bar{\theta}$ defined prior to \autoref{tclt}. Let $\hat\kappa_t = n^{-1} \sum_{i\in\N_n} T_{ti}/p_{ti}$. By \autoref{aboundY} and \autoref{lintersect}, $\tilde Z_i$ is uniformly asymptotically bounded, so 
\begin{equation*}
  \abs{\hat\theta - \tilde\theta} \precsim \abs{\hat\kappa_1^{-1}-1} + \abs{\hat\kappa_0^{-1}-1} \precsim k_n^{-1/2}
\end{equation*}

\noindent by \autoref{lkappa}. It remains to bound the bias and variance of $\tilde\theta$. 

%---------------------
\subsubsection{Bias}
%---------------------

For any $i\in\N_n$, $d_t \in \{0,1\}$, $\bm{d} \in \{0,1\}^n$, and $S \subseteq \N_n$ containing $i$, we use the notation $(d_t, \bm{D}_{S\backslash\{i\}}, \bm{d}_{-S})$ to mean that we take the observed treatment vector $\bm{D}$, replace entry $D_i$ with $d_t$, and replace the subvector $\bm{D}_{\N_n\backslash S} = (D_j\colon j\in\N_n\backslash S)$ with the corresponding entries of $\bm{d}_{\N_n\backslash S}$. 

The basic idea is that if $T_{ti}=1$, then $D_i=d_t$ and $i$ is well surrounded by units belonging to clusters with saturation level $p_t$, so $\E[Y_i \mid T_{ti}=1]$ is a good approximation of $\E^*_{p_t}[Y_i(d_t,\bm{D}_{-i})]$. Formally, 
\begin{align}
  \abs{ \E[Y_i \mid T_{ti}=1] &- \E^*_{p_t}[Y_i(d_t,\bm{D}_{-i})] } \nonumber\\ 
			     &= \lvert \E[Y_i(d_t,\bm{D}_{-i}) \pm Y_i(d_t,\bm{D}_{\N(i,r_n)\backslash\{i\}},\bm{0}_{-\N(i,r_n)}) \mid T_{ti}=1] \nonumber\\ 
			     &- \E^*_{p_t}[Y_i(d_t,\bm{D}_{-i}) \pm Y_i(d_t,\bm{D}_{\N(i,r_n)\backslash\{i\}},\bm{0}_{-\N(i,r_n)})] \rvert \leq 2c\,r_n^{-\gamma}, \label{biasarg}
\end{align}

\noindent because
\begin{equation*}
  \E[Y_i(d_t,\bm{D}_{\N(i,r_n)\backslash\{i\}},\bm{0}_{-\N(i,r_n)}) \mid T_{ti}=1] = \E^*_{p_t}[Y_i(d_t,\bm{D}_{\N(i,r_n)\backslash\{i\}},\bm{0}_{-\N(i,r_n)})]
\end{equation*}

\noindent by \autoref{aCRT} and 
\begin{equation*}
  \abs{Y_i(d_t,\bm{D}_{-i}) - Y_i(d_t,\bm{D}_{\N(i,r_n)\backslash\{i\}},\bm{0}_{-\N(i,r_n)})} \leq c\,r_n^{-\gamma}
\end{equation*}

\noindent by \autoref{aani}. Therefore,
\begin{equation}
  \abs{\E[\tilde\theta] - \theta^*} = \abs{\bar{\theta} - \theta^*} \leq 4c\,r_n^{-\gamma} \sim (k_n\xi_n/n)^{\gamma/d}. \label{fbi}
\end{equation}

\noindent The last part of \eqref{fbi} holds because $r_n$ is the median cluster radius by \eqref{kappa}, and cluster radii are uniformly of order $(n/(k_n\xi_n))^{1/d}$ by \autoref{aclus}. 

%---------------------
\subsubsection{Variance}
%---------------------

Recalling the definition of $\Lambda_i$ from \eqref{Lambdai},
\begin{equation}
  \var\bigg(\frac{1}{n} \sum_{i\in\N_n} \tilde Z_i\bigg) = \underbrace{\frac{1}{n^2} \sum_{i\in\N_n} \sum_{j \in \Lambda_i} \cov(\tilde Z_i,\tilde Z_j)}_{[P1]} + \underbrace{\frac{1}{n^2} \sum_{i\in\N_n} \sum_{j \not\in \Lambda_i} \cov(\tilde Z_i,\tilde Z_j)}_{[P2]}. \label{vararg}
\end{equation}

\noindent By \autoref{lnbhd}, $\max_i \abs{\Lambda_i} \precsim n/k_n$, so $[P1] \precsim k_n^{-1}$. It remains to show that $[P2] \precsim k_n^{-1}$.  Intuitively, $[P2]$ should be well controlled since units not in $\Lambda_i$ are ``far'' from and therefore less correlated with $i$. 

\bigskip
\noindent {\bf Step 1.} Recall that $c(i)$ is the index of the cluster containing unit $i$. Define $X_i^r = \E[\tilde Z_i \mid \F_i(r)]$ for
\begin{equation}
  \F_i(r) = \big\{(D_j,W_{c(j)})\colon \rho(i,j) \leq r\big\}. \label{F_i()}
\end{equation}

\noindent We bound the discrepancy between $\tilde Z_i$ and $X_i^r$. For any $r\geq r_n$ and $t\in\{0,1\}$, $T_{ti}$ is measurable with respect to the $\sigma$-algebra generated by $\F_i(r)$. Then by \autoref{aani}, for any $q>0$, $r\geq r_n$, $n$ sufficiently large, and $t \in \{0,1\}$,
\begin{multline*}
  \E[\abs{\tilde Z_i - X_i^r}^q \mid T_{ti}=1] 
  \\ = p_{ti}^{-q} \E[\abs{Y_i(\bm{D}) - \E[Y_i(\bm{D}) \mid \F_i(r)]}^q \mid T_{ti}=1] \leq (p_{ti}^{-1} 2c\,r^{-\gamma})^q
\end{multline*}

\noindent using an argument similar to \eqref{biasarg}. By \autoref{lintersect}, $\max_i p_{ti}^{-1}\precsim 1$, so by the law of total probability, there exists $c_n \precsim 1$ such that 
\begin{equation}
  \max_{i\in\N_n} \E[\abs{\tilde Z_i - X_i^r}^q]^{1/q} \leq c_nr^{-\gamma} \quad\text{for all}\quad r\geq r_n. \label{e1290urt2g0}
\end{equation}

\noindent {\bf Step 2.} Fix $i,j\in\N_n$ such that $j\not\in\Lambda_i$. We use \eqref{e1290urt2g0} to bound $\cov(\tilde Z_i, \tilde Z_j)$. Since the $r_n$-neighborhoods of $i$ and $j$ do not intersect a common cluster, $X_i^{r_n} \indep X_j^{r_n}$ by \autoref{aCRT}. Applying the Cauchy-Schwarz inequality, \eqref{e1290urt2g0} for $q=2$, and \autoref{aboundY} and \autoref{lintersect}, there exists $c_n' \precsim 1$ such that for $n$ sufficiently large and all $i,j$ such that $j\not\in\Lambda_i$,
\begin{multline}
  \abs{\cov(\tilde Z_i,\tilde Z_j)} \leq \abs{\cov(X_i^{r_n}, X_j^{r_n})} + \abs{\cov(X_i,\tilde Z_j - X_j^{r_n})} \\ + \abs{\cov(\tilde Z_i - X_i^{r_n},X_j)} + \abs{\cov(\tilde Z_i - X_i^{r_n}, \tilde Z_j - X_j^{r_n})} \leq c_n' r_n^{-\gamma}. \label{g902jghbwa}
\end{multline}

Now suppose additionally that $\rho(i,j) > 4\bar{R}$ for 
\begin{equation*}
  \bar{R} = \max_j R_j,
\end{equation*}

\noindent where $R_j$ is defined prior to \eqref{kappa}. In this case, we derive a different covariance bound. Since $\N(i,\rho(i,j)/2-\bar{R})$ and $\N(j,\rho(i,j)/2-\bar{R}))$ are separated by a distance of at least $2\bar{R}$, which upper bounds the ``diameter'' of any cluster, we have $X_i^{\rho(i,j)/2-\bar{R}} \indep X_j^{\rho(i,j)/2-\bar{R}}$. By a derivation similar to \eqref{g902jghbwa} using $\rho(i,j)/2-\bar{R}$ in place of $r_n$,
\begin{equation}
  \abs{\cov(\tilde Z_i,\tilde Z_j)} \ind\{\rho(i,j) > 4\bar{R}\} \leq c_n'(\rho(i,j)/2-\bar{R})^{-\gamma} \ind\{\rho(i,j) > 4\bar{R}\}. \label{covzz0}
\end{equation}

\noindent {\bf Step 3.} Let $\lceil c \rceil$ ($\lfloor c \rfloor$) denote $c$ rounded up (down) to the nearest integer. Using the covariance bounds derived in step 2, 
\begin{multline}
  \frac{1}{n^2} \sum_{i\in\N_n} \sum_{j \not\in \Lambda_i} \abs{\cov(\tilde Z_i,\tilde Z_j)} \leq \frac{c_n'}{n^2} \sum_{\ell=1}^{k_n} \sum_{i\in C_\ell} \sum_{s=\lfloor 2r_n \rfloor}^\infty \sum_{j\not\in\Lambda_i} \bm{1}\{\rho(i,j)\in [s,s+1)\} \\ \times \left( r_n^{-\gamma} \bm{1}\{s \leq 4\bar{R}\} + (s/2-\bar{R})^{-\gamma} \bm{1}\{s > 4\bar{R}\} \right) \equiv [P2.1] + [P2.2], \label{covzz}
\end{multline} 

\noindent where $[P2.1]$ takes the part involving $s \leq 4\bar{R}$ and $[P2.2]$ the part involving $s > 4\bar{R}$. The sum over $s$ starts at $\lfloor 2r_n \rfloor$ because $j\not\in\Lambda_i$ implies that the $r_n$ neighborhoods of $i,j$ do not intersect. 

We next show that $\eqref{covzz} \precsim \xi_n/n$. By Assumptions \ref{aspace}(a) and \ref{aclus}, $\max_\ell \abs{C_\ell} \precsim n/k_n$. By \autoref{aspace}(b), $\sum_{j\not\in\Lambda_i} \bm{1}\{\rho(i,j)\in [s,s+1)\} \leq \abs{\N(i,s+1)\backslash\N(i,s)} \leq C \xi_n \max\{s^{d-1},1\}$. Then
\begin{equation}
  [P2.1] = c_n' \frac{k_n}{n^2} \frac{n}{k_n} \sum_{s=\lfloor 2r_n \rfloor}^{\lceil 4\bar{R} \rceil} C\xi_n s^{d-1} r_n^{-\gamma} \precsim \frac{\xi_n}{n} \bar{R}^d r_n^{-\gamma} \label{varP2.1}
\end{equation}

\noindent because
\begin{equation*}
  \sum_{s=\lfloor 2r_n \rfloor}^{\lceil 4\bar{R} \rceil} s^{d-1} \leq \int_{\lfloor 2r_n \rfloor}^{\lceil 4\bar{R} \rceil + 1} s^{d-1} \,\text{d}s = d^{-1}\left( \lceil 4\bar{R}+1 \rceil^d - \lfloor 2r_n \rfloor^d \right) \precsim \bar{R}^d
\end{equation*}
% For f increasing, f(n) \leq \int_n^{n+1} f(x) dx.
% Sum the LHS from n=t to u to get \sum_{n=t}^u f(n) \leq \int_t^{u+1} f(x) dx

\noindent given $\bar{R} \geq r_n$ by definition. By \autoref{aclus},  
\begin{equation}
  r_n \geq L_n \sim (n/(k_n\xi_n))^{1/d},
  \label{r_nbound}
\end{equation}

\noindent and $\bar{R}^d \precsim n/(k_n\xi_n)$, so
\begin{equation}
  \eqref{varP2.1} \precsim \frac{\xi_n}{n} \left( \frac{n}{k_n\xi_n} \right)^{1-\gamma/d} = \frac{1}{k_n} \left( \frac{k_n\xi_n}{n} \right)^{\gamma/d} \precsim \frac{\xi_n}{n} \label{varP2.15}
\end{equation}

\noindent since $\gamma>d$ by \autoref{aani}. Similarly, using the fact that $d \geq 1$ from \autoref{aspace}(b),
\begin{multline}
  [P2.2] = c_n' \frac{k_n}{n^2} \frac{n}{k_n} \sum_{s=\lfloor 4\bar{R} \rfloor}^\infty C\xi_n s^{d-1} (s/2 - \bar{R})^{-\gamma} \\ 
  \leq C c_n' \frac{\xi_n}{n} \sum_{s=\lfloor 4\bar{R} \rfloor}^\infty s^{d-1} (s/4)^{-\gamma} \precsim \frac{\xi_n}{n} \sum_{s=1}^\infty s^{d-1-\gamma}. \label{varP2.2}
\end{multline}

\noindent Combining \eqref{covzz}, \eqref{varP2.1}, \eqref{varP2.15}, and \eqref{varP2.2} yields $[P2] \precsim \xi_n/n \precsim k_n^{-1}$ since $k_n\xi_n/n \precsim 1$ by assumption. \qed

%---------------------------------------
\subsection{\autoref{tplus}}\label{stplus}
%---------------------------------------

%---------------------
\subsubsection{Part (a)}
%---------------------

Let $T_{ti}^+ = \ind\{D_i=d_t, W_{c(i)}=t\}$ for $d_t \in \{0,1\}$ and $T_{ti}^+ = \ind\{W_{c(i)}=t\}$ for $d_t = \emptyset$. Define $p_t^+ = \E[T_{ti}^+]$, $\hat p_t^+ = n^{-1} \sum_{i\in\N_n} T_{ti}^+$, and
\begin{equation*}
  \tilde\theta^+ = \frac{1}{n} \sum_{i\in\N_n} \left( \frac{T_{1i}^+}{p_1^+} - \frac{T_{0i}^+}{p_0^+} \right) Y_i,
\end{equation*}

\noindent the Horvitz-Thompson analog of the difference in means $\hat\theta^+$. We have
\begin{equation*}
  \hat\theta^+ - \tilde\theta^+ = \frac{p_1^+-\hat{p}_1^+}{\hat{p}_1^+} \frac{1}{n} \sum_{i\in\N_n} \frac{T_{1i}^+Y_i}{p_1^+} - \frac{p_0^+-\hat{p}_0^+}{\hat{p}_0^+} \frac{1}{n} \sum_{i\in\N_n} \frac{T_{0i}^+Y_i}{p_0^+} \precsim k_n^{-1/2}
\end{equation*}

\noindent because $T_{ti}^+Y_i/p_t^+$ is uniformly asymptotically bounded by \autoref{aboundY} and \autoref{lintersect}, and $p_t^+ - \hat{p}_t^+ \precsim k_n^{-1/2}$ by \autoref{lpt}. It suffices to derive the rate of convergence of $\tilde\theta^+$. 

By the argument used to bound \eqref{vararg} in the proof of \autoref{tQ}, $\var(\tilde\theta^+) \precsim k_n^{-1}$. It remains to bound the bias, that is, to show that
\begin{equation}
  \abs{\E[\tilde\theta^+] - \theta^*} \precsim (k_n\xi_n/n)^{1/d}. \label{QAB}
\end{equation}

\noindent Recall the definitions in \autoref{aclus}, letting $m_j$ denote the ``centroid'' of cluster $C_j$. Define $j$'s ``boundary'' as
\begin{equation*}
  \mathcal{B}(C_j) = \N(m_j,U_n) \backslash \N(m_j,L_n). 
\end{equation*}

\noindent By \autoref{aclus}, there exists $\alpha_n \precsim 1$ such that $U_n - L_n = \alpha_n (n/(k_n\xi_n))^{1/d}$, so by \autoref{aspace}(b), 
\begin{equation}
  \max_j \abs{\mathcal{B}(C_j)} \leq \sum_{\ell=0}^{\lceil \alpha_n \rceil} \abs{\N(m_j,L_n+\ell+1) \backslash \N(m_j,L_n+\ell)} \precsim \xi_n (n/(k_n\xi_n))^{(d-1)/d}.
  \label{precboundary}
\end{equation}

For $r = 0, \ldots, L_n$, define the ``contour sets''
\begin{equation*}
  J(r,C_j) = \N(m_j,L_n-r) \backslash \N(m_j,L_n-r-1), 
\end{equation*}

\noindent where $\N(i,-1) \equiv \emptyset$. As $r$ increases, $J(r,C_j)$ moves away from the boundary and towards the interior. If $i \in J(r,C_j)$, then $\N(i,r) \subseteq C_j$, so for such $i$, by Assumptions \ref{aani} and \ref{aCRT},
\begin{align*}
  \big| (p_t^+)^{-1}\E[Y_iT_{ti}^+] 
  &- \E^*_{p_t}[Y_i(d_t,\bm{D}_{-i})] \big| \\
  &= \big| \E[Y_i(d_t,\bm{D}_{-i}) \pm Y_i(d_t,\bm{D}_{\N(i,r)\backslash\{i\}},\bm{0}_{-\N(i,r)}) \mid T_{ti}^+=1] \\ 
  &- \E^*_{p_t}[Y_i(d_t,\bm{D}_{-i}) \pm Y_i(d_t,\bm{D}_{\N(i,r)\backslash\{i\}},\bm{0}_{-\N(i,r)})] \big| \\
  &\leq 2c\,\min\{r^{-\gamma},1\}
\end{align*}

\noindent for any $t\in\{0,1\}$. Then
\begin{align*}
  \frac{1}{n} \sum_{i\in\N_n} \big| 
  &(p_t^+)^{-1}\E[Y_iT_{ti}^+] - \E^*_{p_t}[Y_i(d_t,\bm{D}_{-i})] \big| \\
  &\precsim \frac{1}{n} \sum_{j=1}^{k_n} \left( \abs{\mathcal{B}(C_j)} + \sum_{r=0}^{L_n} \sum_{i \in J(r,C_j)} \min\{r^{-\gamma},1\} \right) \\
  &\precsim \underbrace{\frac{k_n}{n} \xi_n \left( \frac{n}{k_n\xi_n} \right)^{\frac{d-1}{d}}}_\text{boundary bias} + \frac{k_n}{n} \sum_{r=0}^{L_n} \xi_n\max\{(L_n-r-1)^{d-1},1\} \min\{r^{-\gamma},1\} \\
  &\precsim (k_n\xi_n/n)^{1/d} + \frac{k_n\xi_n}{n} L_n^{d-1}\sum_{r=0}^\infty \min\{r^{-\gamma},1\}
  \precsim (k_n\xi_n/n)^{1/d}.
\end{align*}

\noindent The second line uses \autoref{aboundY} and \autoref{lintersect}. It converts the sum over all units to a sum over all clusters followed by sums over units in each contour set through the boundary. The third line bound on $\abs{\mathcal{B}(C_j)}$ follows from \eqref{precboundary}, while the bound on $\abs{J(r,C_j)}$ uses \autoref{aspace}(b). This establishes \eqref{QAB}.

%---------------------
\subsubsection{Part (b)}
%---------------------

Let $d\geq 1$, and fix any sequence $\{\xi_n\}_{n\in\mathbb{N}}$ such that $1 \precsim \xi_n \prec n$. Let $\rho$ be the sup norm and $\N_n = \{\xi_n^{-1/d}x\colon x \in \mathbb{Z}^d\} \cap B(\bm{0},\mathcal{R}_n)$ where $B(\bm{0},\mathcal{R}_n)$ is a (hyper)cube with radius $\mathcal{R}_n \in \mathbb{Z}$ and $\mathcal{R}_n \sim (n/\xi_n)^{1/d}$. When $\xi_n\sim 1$, the observed units are positioned within a cube of radius $\sim n^{1/d}$ containing $\sim n$ units positioned on the integer lattice. When $\xi_n \succ 1$, we shrink this region towards the origin by a factor $\xi_n^{-1/d}$. This results in the same asymptotic order of number of units $\abs{\N_n} \sim n$, but the volume of the cube is reduced to $\sim \mathcal{R}_n^d \sim n/\xi_n$, resulting in a density of $\xi_n$.

\bigskip
\noindent {\bf Verifying \autoref{aspace}.} By construction of $\N_n$, there exist $C_0,C_1>0$ such that for $r\geq 0$ and $i\in\N_n$, 
\begin{align*}
  &\abs{\N(i,\xi_n^{-1/d}r)} \in [\min\{C_0^{-1}r^d,n\}, \max\{C_0r^d,1\}] \quad\text{and} \\
  &\abs{\N(i,\xi_n^{-1/d}(r+1))\backslash \N(i,\xi_n^{-1/d}r)} \in [\min\{C_1^{-1}r^{d-1},n\}, \max\{C_1r^{d-1},1\}].
\end{align*}

\noindent Substituting $\xi_n^{1/d}r$ for $r$ in these expressions,
\begin{align}
  &\abs{\N(i,r)} \in [\min\{C_0^{-1}\xi_nr^d,n\}, \max\{C_0\xi_nr^d,1\}] \quad\text{and} \label{csz}\\
  &\abs{\N(i,r+\xi_n^{-1/d})\backslash \N(i,r)} \in [\min\{C_1^{-1}\xi_n^{(d-1)/d}r^{d-1},n\}, \max\{C_1\xi_n^{(d-1)/d}r^{d-1},1\}]. \nonumber
\end{align}

\noindent Letting $[c]$ denote rounding $c$ to the nearest integer, this implies
\begin{align}
  \abs{\N(i,r+1)\backslash \N(i,r)} 
  &\leq 1 + \sum_{k=1}^{[\xi_n^{1/d}]} \abs{\N(i,r+\xi_n^{-1/d}k)\backslash \N(i,r+(\xi_n^{-1/d}(k-1))} \nonumber\\
  &\leq 1 + \sum_{k=1}^{[\xi_n^{1/d}]} C_1 \xi_n^{(d-1)/d} (r+\xi_n^{-1/d}(k-1))^{d-1} \nonumber\\
  &\leq 1 + C_1 \xi_n^{(d-1)/d} \int_1^{[\xi_n^{1/d}]+1} (r+\xi_n^{-1/d}(k-1))^{d-1}\,\text{d}k \nonumber\\
  &= 1 + C_1 \xi_n^{(d-1)/d} \left( d^{-1} \xi_n^{1/d} (r+\xi_n^{-1/d}(k-1))^d  \right) \bigg|_1^{[\xi_n^{1/d}]+1} \nonumber\\
  &\leq 1 + C_1 d^{-1} \xi_n ((r+\xi_n^{-1/d}[\xi_n^{1/d}])^d - r^d), \label{bdry}
\end{align}

\noindent which is bounded by a constant times $\xi_n r^{d-1}$. This verifies \autoref{aspace}. 

\bigskip
\noindent {\bf Verifying \autoref{aclus}.} Let $1 \prec k_n \prec n/\xi_n$. For the remainder of the proof, suppose the clusters are generated by $k_n$-medoids which by \autoref{tkmed} satisfy \autoref{aclus}. Also let $n$ be sufficiently large that $L_n$ in the assumption exceeds two.

\bigskip
\noindent {\bf Verifying Assumptions \ref{aboundY} and \ref{aani}.} Set
\begin{equation*}
  Y_i(\bm{d}) = d_i\frac{\sum_{j\in\mathcal{N}_n} d_j \bm{1}\{\rho(i,j)\leq 2\}}{\sum_{j\in\mathcal{N}_n} \bm{1}\{\rho(i,j)\leq 2\}},
\end{equation*}

\noindent which is unit $i$'s treatment times the fraction of treated units in $i$'s {\em 2-neighborhood}. The sum in the denominator is always at least one since it includes $i$ itself, so \autoref{aboundY} holds. \autoref{aani} holds because potential outcomes only depend on 2-neighborhood treatments.

\bigskip
\noindent {\bf Reduction to HT.} Notice $Y_i(\bm{0}) = 0$. Since $Q \in \{D,T,O\}$ and $p_0=0$, $T_{0i}^+Y_i = 0$. Then 
\begin{equation*}
  \hat\theta^+ = \underbrace{\frac{1}{n} \sum_{i\in\N_n} \frac{T_{1i}^+}{p_1^+} Y_i}_{\tilde{\theta}^+} \frac{p_1^+}{\hat{p}_1^+}. 
\end{equation*}

\noindent By \autoref{lpt}, $p_1^+/\hat{p}_1^+ \plimarrow 1$, so it remains to lower bound the bias and variance of $\tilde\theta^+$.

\bigskip
\noindent {\bf Bias.} Let $\beta = (p_1-p_1q) \ind\{d_1=1\} + p_1 (p_1-p_1q) \ind\{d_1=\emptyset\}$, which is strictly positive by assumption. We have
\begin{multline}
  \abs{\E[\tilde\theta^+] - \theta^*} = \bigg| \frac{1}{n} \sum_{i\in\N_n} \big( \E[Y_i \mid T_{1i}^+=1] - \E^*_{p_1}[Y_i(d_1,\bm{D}_{-i})] \big) \bigg| \\
  = \beta \frac{1}{n} \sum_{i\in\N_n} \frac{\sum_{j \in \N_n\backslash C_{c(i)}} \bm{1}\{\rho(i,j)\leq 2\}}{\sum_{j\in\mathcal{N}_n} \bm{1}\{\rho(i,j)\leq 2\}}. \label{lbias}
\end{multline}

\noindent Define the boundary of a set $C \subseteq \N_n$ as
\begin{equation*}
  \tilde{\mathcal{B}}(C) = \big\{i \in C\colon \abs{\N(i,1) \cap (\{\xi_n^{-1/d}x \colon x \in \mathbb{Z}^d\}\backslash C)} \geq 1\big\},
\end{equation*}

\noindent the set of units whose $1$-neighborhoods intersect a point outside the set. Then
\begin{equation}
  \eqref{lbias} \geq \beta \frac{1}{n} \sum_{j=1}^{k_n} \sum_{i\in\tilde{\mathcal{B}}(C_j)\backslash\tilde{\mathcal{B}}(\N_n)} \frac{\sum_{\ell \in \N_n\backslash C_j} \bm{1}\{\rho(i,\ell)\leq 2\}}{\sum_{\ell\in\mathcal{N}_n} \bm{1}\{\rho(i,\ell)\leq 2\}}. \label{lbias2}
\end{equation}

\noindent By construction of $\N_n$,
\begin{equation*}
  \min_{j=1,\ldots,k_n} \min_{i\in\tilde{\mathcal{B}}(C_j)\backslash\tilde{\mathcal{B}}(\N_n)} \frac{\sum_{\ell \in \N_n\backslash C_j} \bm{1}\{\rho(i,\ell)\leq 2\}}{\sum_{\ell\in\mathcal{N}_n} \bm{1}\{\rho(i,\ell)\leq 2\}} 
  \succsim 1.
\end{equation*} 

\noindent Therefore,
\begin{equation}
  \eqref{lbias2} \succsim \frac{1}{n} \bigg( \sum_{j=1}^{k_n} \abs{\tilde{\mathcal{B}}(C_j)} - \abs{\tilde{\mathcal{B}}(\N_n)} \bigg). \label{lbias2.5}
\end{equation}

\noindent By \eqref{bdry}, 
\begin{equation}
  \abs{\tilde{\mathcal{B}}(\N_n)} \sim \xi_n \mathcal{R}_n^{d-1} \sim n (\xi_n/n)^{1/d}
  \label{areabdry}
\end{equation}

Let $m_j$ denote the medoid of a cluster $C_j$. By \autoref{aclus}(a), units in $\tilde{\mathcal{B}}(C_j)$ are at least distance $\lfloor L_n \rfloor$ from $m_j$. Define the ``contour set''
\begin{equation*}
  \tilde J(r,C_j) = \N(m_j,r)\backslash\N(m_j,r-1).
\end{equation*}

\noindent By construction of $\N_n$, $\abs{\tilde J(r,C_j)}$ is increasing in $r$, so 
\begin{equation*}
  \abs{\tilde{\mathcal{B}}(C_j)} \geq \abs{\tilde J(\lfloor L_n \rfloor,C_j)} \sim \xi_n (n/(k_n\xi_n))^{(d-1)/d}
\end{equation*}

\noindent uniformly in $j$ by \eqref{bdry} and \autoref{aclus}(b). Hence
\begin{equation*}
  \frac{1}{n} \sum_{j=1}^{k_n} \abs{\tilde{\mathcal{B}}(C_j)} \succsim \frac{k_n\xi_n}{n} \left( \frac{n}{k_n\xi_n} \right)^{(d-1)/d} = (k_n\xi_n/n)^{1/d}.
\end{equation*}

\noindent Combined with \eqref{areabdry}, $\eqref{lbias2.5} \succsim (k_n\xi_n/n)^{1/d}$.

\bigskip
\noindent {\bf Variance.} For $Z_i^+ = T_{1i}^+ Y_i/p_1^+$, $\cov(Z_i^+,Z_j^+)$ equals
\begin{multline*}
  \sum_{\ell\in\N_n} \sum_{m\in\N_n} \frac{\bm{1}\{\rho(i,\ell)\leq 2\}}{\sum_{\ell'\in\N_n} \bm{1}\{\rho(i,\ell')\leq 2\}} \frac{\bm{1}\{\rho(j,m)\leq 2\}}{\sum_{m'\in\N_n} \bm{1}\{\rho(j,m')\leq 2\}} \\
  \times (p_1^+)^{-2} \cov\big( T_{1i}^+ D_i D_\ell, T_{1j}^+ D_j D_m \big).
\end{multline*}

\noindent Since $Q\in\{D,T,O\}$, either $T_{1i}^+ = D_iW_{c(i)}$ in the case of $Q \in \{D,T\}$ or $T_{1i}^+ = W_{c(i)}$ in the case of $Q=O$, so the covariance term equals $\cov(W_{c(i)}D_iD_\ell, W_{c(j)}D_jD_m) \geq 0$. Furthermore, for $j \in C_{c(i)}$, $W_{c(i)}=W_{c(j)}$, so
\begin{multline*}
  \cov(W_{c(i)}D_iD_\ell, W_{c(j)}D_jD_m) \\ = \underbrace{\E[W_{c(i)} \cov(D_iD_\ell, D_jD_m \mid W_{c(i)})]}_{\geq 0} + \cov(\E[D_iD_\ell \mid W_{c(i)}], \E[D_jD_m \mid W_{c(i)}]).
\end{multline*}

\noindent If additionally $\ell \in C_{c(i)}\backslash\{i\}$ and $m\in C_{c(i)}\backslash\{j\}$, the covariance term on the right-hand side equals $\alpha \equiv qp_1^4 - (qp_1^2)^2 > 0$. 
% If $\ell \in C_{c(i)}\backslash\{i\}$ and $m\not\in C_{c(i)}$, it equals $q(p_1^2 \cdot p_1 \cdot qp_1) - qp_1^2 \cdot q(p_1 \cdot qp_1)>0$. 
Then by construction of $\N_n$,
\begin{multline*}
  \var\left( \frac{1}{n} \sum_{i\in\N_n} Z_i^+ \right) \geq \alpha \frac{1}{n^2} \sum_{i\in\N_n} \sum_{j\in C_{c(i)}} \sum_{\ell\in C_{c(i)}\backslash\{i\}} \frac{\bm{1}\{\rho(i,\ell)\leq 2\}}{\sum_{\ell'\in\N_n} \bm{1}\{\rho(i,\ell')\leq 2\}} \\
  \times \sum_{m\in C_{c(i)}\backslash\{j\}} \frac{\bm{1}\{\rho(j,m)\leq 2\}}{\sum_{m'\in\N_n} \bm{1}\{\rho(j,m')\leq 2\}} \succsim \frac{1}{n^2} \sum_{i\in\N_n} \abs{C_{c(i)}}.
\end{multline*}

\noindent By \eqref{csz} and \autoref{aclus}(a), $\min_j \abs{C_j} \geq \abs{\N(m_j,L_n)} \succsim n/k_n$. Hence, the right-hand side of the above display is at least order $k_n^{-1}$. \qed

%---------------------------------------
\subsection{\autoref{tclt}}\label{stclt}
%---------------------------------------

The first three steps establish \eqref{cltsec}, and step four proves \eqref{cltmain}. 

\bigskip
\noindent {\bf Step 1.} We first derive an asymptotically linear representation. For $\hat\kappa_t = n^{-1} \sum_{i\in\N_n} T_{ti} / p_{ti}$ and $\mu_t = n^{-1} \sum_{i\in\N_n} \E[Y_i \mid T_{ti}=1]$,
\begin{equation*}
  \frac{\sum_{i\in\N_n} Y_i T_{ti} / p_{ti}}{\sum_{i\in\N_n} T_{ti} / p_{ti}} - \mu_t = \hat\kappa_t^{-1} \frac{1}{n} \sum_{i\in\N_n} \frac{T_{ti}(Y_i - \mu_t)}{p_{ti}}.
\end{equation*}

\noindent By \autoref{lkappa}, $\hat\kappa_t - 1 \precsim k_n^{-1/2}$, so 
\begin{equation}
  \hat\theta - \bar{\theta} = \frac{1}{n} \sum_{i\in\N_n} \underbrace{\left( \frac{T_{1i}(Y_i - \mu_1)}{p_{1i}} - \frac{T_{0i}(Y_i - \mu_0)}{p_{0i}} \right)}_{Z_i} + o_p(k_n^{-1/2}) \label{Z_iQ}
\end{equation}

\noindent since $n^{-1} \sum_{i\in\N_n} p_{ti}^{-1} T_{ti}(Y_i - \mu_t) \prec 1$ by the variance calculation in the proof of \autoref{tQ}. 

\bigskip
\noindent {\bf Step 2.} Recall the definition of $\F_i(r)$ from \eqref{F_i()} and $\Lambda_i$ from \eqref{Lambdai}. Define $B_i = Z_i - \E[Z_i \mid \F_i(r_n)]$. We next establish that
\begin{equation*}
  \E\left[ \left( \sqrt{k_n} \frac{1}{n} \sum_{i\in\N_n} B_i \right)^2 \right] \precsim \frac{k_n\xi_n}{n} \prec 1,
\end{equation*}

\noindent where the last asymptotic inequality follows from the assumption $k_n \prec n/\xi_n$. Expanding the square yields
\begin{equation*}
  \frac{k_n}{n^2} \sum_{i\in\N_n} \sum_{j\in\Lambda_i} \E[B_iB_j] + \frac{k_n}{n^2} \sum_{i\in\N_n} \sum_{j\not\in\Lambda_i} \E[B_iB_j] \equiv [P1] + [P2].
\end{equation*}

For all $r\geq r_n$ and $t\in\{0,1\}$, $T_{ti}$ is measurable with respect to $\F_i(r)$, so
\begin{multline*}
  Z_i-\E[Z_i \mid \F_i(r)] \\ = \frac{T_{1i}}{p_{1i}} (Y_i(\bm{D}) - \E[Y_i(\bm{D}) \mid \F_i(r)]) - \frac{T_{0i}}{p_{0i}} (Y_i(\bm{D}) - \E[Y_i(\bm{D}) \mid \F_i(r)]).
\end{multline*}

\noindent By \autoref{aani},
\begin{multline*}
  \lvert Y_i(\bm{D}) - \E[Y_i(\bm{D}) \mid \F_i(r)] \rvert = \lvert Y_i(\bm{D}) \pm Y_i(\bm{D}_{\N(i,r)}, \bm{0}_{-\N(i,r)}) \\
  - \E[Y_i(\bm{D}) \pm Y_i(\bm{D}_{\N(i,r)}, \bm{0}_{-\N(i,r)}) \mid \F_i(r)] \rvert \leq 2c \min\{r^{-\gamma}, 1\}
\end{multline*}

\noindent since $\E[Y_i(\bm{D}_{\N(i,r)}, \bm{0}_{-\N(i,r)}) \mid \F_i(r)] = Y_i(\bm{D}_{\N(i,r)}, \bm{0}_{-\N(i,r)})$. Then by \autoref{lintersect}, there exists $c_n \precsim 1$ such that for all $i$, 
\begin{equation}
  \abs{Z_i-\E[Z_i \mid \F_i(r)]} \leq c_n \min\{r^{-\gamma},1\} \quad\text{for all}\quad r\geq r_n. \label{bhn0ejr5}
\end{equation}

\noindent By \autoref{lnbhd}, $\max_i \abs{\Lambda_i} \precsim n/k_n$, so \eqref{r_nbound} and \eqref{bhn0ejr5} imply
\begin{equation*}
  \abs{[P1]} \precsim (c_n r_n^{-\gamma})^2 \precsim (k_n\xi_n/n)^{2\gamma/d},
\end{equation*}

\noindent which is $\precsim k_n\xi_n/n$ since $\gamma>d$ by \autoref{aani}.

Turning to $\abs{[P2]}$, fix $i,j$ such that $j\not\in\Lambda_i$. By \eqref{bhn0ejr5}, there exists $c_n \precsim 1$ such that for all such $i,j$,
\begin{equation*}
  \abs{\E[B_iB_j]} \leq c_n r_n^{-\gamma}.
\end{equation*} 

\noindent We will use a different bound when additionally $\rho(i,j) > 4\bar{R}$ for $\bar{R} = \max_j R_j$. In this case we have
\begin{equation*}
  B_i - \underbrace{\E[B_i \mid \F_i(\rho(i,j)/2-\bar{R})]}_{X_i} = Z_i - \E[Z_i \mid \F_i(\rho(i,j)/2-\bar{R})] 
\end{equation*}

\noindent since $\N(i,\rho(i,j)/2-\bar{R}) \supseteq \N(i,r_n)$. Notice $\E[X_i]=0$; $\max_i \abs{X_i} \precsim 1$ by \autoref{aboundY} and \autoref{lintersect}; and $X_i \indep X_j$ by \autoref{aCRT} since $\N(i,\rho(i,j)/2-\bar{R})$ and $\N(j,\rho(i,j)/2-\bar{R})) > 2\bar{R}$ are separated by a distance of at least $2\bar{R}$, which is an upper bound on the ``diameter'' of any cluster. Applying \eqref{bhn0ejr5} with $\rho(i,j)/2-\bar{R}$ in place of $r_n$, there exists $c_n' \precsim 1$ such that for any $i,j$,
\begin{multline*}
  \abs{\E[B_iB_j]} \leq \abs{\E[X_iX_j]} + \abs{\E[X_i(B_j-X_j)]} \\ + \abs{\E[(B_i-X_i)X_j]} + \abs{\E[(B_i-X_i)(B_j-X_j)]} \leq c_n'(\rho(i,j)/2-\bar{R})^{-\gamma}.
\end{multline*}

\noindent These bounds yield
\begin{multline*}
  \abs{[P2]} \leq \frac{k_n}{n^2} \sum_{i\in\N_n} \sum_{j \not\in \Lambda_i} \abs{\E[B_iB_j]} \\ \leq k_n \frac{c_n'}{n^2} \sum_{\ell=1}^{k_n} \sum_{i\in C_\ell} \sum_{s=\lfloor 2r_n \rfloor}^\infty \sum_{j\not\in\Lambda_i} \bm{1}\{\rho(i,j)\in [s,s+1)\} \\
  \times \big( r_n^{-\gamma} \bm{1}\{s \leq 4\bar{R}\} + (s/2-\bar{R})^{-\gamma} \bm{1}\{s > 4\bar{R}\} \big).
\end{multline*}

\noindent The right-hand side equals \eqref{covzz} multiplied by $k_n$. Since $\eqref{covzz} \precsim \xi_n/n$ as shown in the proof of \autoref{tclt}, this is $\precsim k_n\xi_n/n$, as desired. 

\bigskip
\noindent {\bf Step 3.} Letting $\tilde\sigma_n^2 = \var(\sqrt{k_n}n^{-1} \sum_{i\in\N_n} \E[Z_i \mid \F_i(r_n)])$, 
\begin{equation}
  \abs{\sigma_n - \tilde\sigma_n} \leq \var\left( \sqrt{k_n} \frac{1}{n} \sum_{i\in\N_n} (Z_i - \E[Z_i \mid \F_i(r_n)]) \right)^{1/2} \prec 1 \label{mink}
\end{equation}

\noindent by Minkowski's inequality and step 2. Since $\sigma_n^2 \succsim 1$ by assumption, it suffices to show
\begin{equation}
  \tilde\sigma_n^{-1} \sqrt{k_n} \frac{1}{n} \sum_{i\in\N_n} \big( \E[Z_i \mid \F_i(r_n)] - \E[Z_i] \big) \dlimarrow \N(0,1) \label{dpclt}
\end{equation}

\noindent to establish \eqref{cltsec}. We apply Theorem 3.6 of \cite{ross2011fundamentals}, defining his $X_i$ as $n^{-1}k_n^{1/2} (\E[Z_i \mid \F_i(r_n)] - \E[Z_i])$ and his dependency graph $\bm{A}$ by connecting units $i,j$ in $\bm{A}$ if and only if $j \in \Lambda_i$. This is a dependency graph because $j \not\in\Lambda_i$ implies that the treatment assignments determining $\F_i(r_n)$ are independent of those determining $\F_j(r_n)$ under \autoref{aCRT}. The maximum degree of $\bm{A}$ is at most $\max_i \abs{\Lambda_i} \precsim n/k_n$ by \autoref{lnbhd}. Therefore, by \autoref{aboundY}, the right-hand side of (3.8) in \cite{ross2011fundamentals} is asymptotically bounded above by
\begin{equation*}
  \left( \frac{n}{k_n} \right)^2 n \left( \frac{\sqrt{k_n}}{n} \right)^3 + \left( \frac{n}{k_n} \right)^{3/2} \sqrt{n\left( \frac{\sqrt{k_n}}{n} \right)^4} \prec 1,
\end{equation*}

\noindent so \eqref{dpclt} follows from his (3.8). This completes the proof of \eqref{cltsec}

\bigskip
\noindent {\bf Step 4.} Note that $\bar{\theta} = \E[\tilde\theta]$, where $\tilde\theta$ is the Horvitz-Thompson analog of $\hat\theta$ defined in \eqref{HT}. The bias of $\tilde\theta$ is $\abs{\bar{\theta} - \theta^*} \precsim (k_n\xi_n/n)^{\gamma/d}$ by \eqref{fbi}. Given that $k_n \prec (n/\xi_n)^{\frac{2\gamma}{2\gamma+d}}$, 
\begin{equation}
  \sqrt{k_n}(\bar{\theta} - \theta^*) \precsim \sqrt{k_n}(k_n\xi_n/n)^{\gamma/d} \prec \sqrt{k_n}(n/\xi_n)^{-\frac{\gamma}{2\gamma+d}} \prec 1, \label{tggd}
\end{equation}

\noindent in which case
\begin{equation*}
  \sigma_n^{-1} \sqrt{k_n} (\hat\theta - \bar{\theta}) = \sigma_n^{-1} \sqrt{k_n} (\hat\theta - \theta^*) + o(1),
\end{equation*}

\noindent so \eqref{cltmain} follows from \eqref{cltsec}. \qed

%---------------------------------------
\subsection{\autoref{tvar}}\label{stvar}
%---------------------------------------

In what follows, steps 1--4 concern case $\hat\sigma^2(1)$, and step 5 concerns $\hat\sigma^2(2)$. Define $Z_i$ as in \eqref{Z_iQ} and $\hat{Z}_i$ as in \eqref{hatsigma}. Let
\begin{equation*}
  \mathcal{B}_n = \frac{k_n}{n^2} \sum_{i\in\N_n} \sum_{j\in\N_n} \E[Z_i] \E[Z_j] A_{ij}(2).
\end{equation*}

\noindent This equals \eqref{Bn}, which is non-negative for any $n$. 

By \eqref{mink}, $\abs{\sigma_n - \tilde\sigma_n} \prec 1$, where 
\begin{equation}
  \tilde\sigma_n^2 = \frac{k_n}{n^2} \sum_{i\in\N_n} \sum_{j\in\N_n} \cov(\E[Z_i \mid \F_i(r_n)], \E[Z_j \mid \F_j(r_n)]) A_{ij}(1) \label{tildesigma}
\end{equation}

\noindent by \autoref{aCRT}. Thus, to show that $\hat\sigma^2(1) = \sigma_n^2 + \mathcal{B}_n + o_p(1)$, it suffices to prove
\begin{equation}
  \hat\sigma^2(1) = \tilde\sigma_n^2 + \mathcal{B}_n + o_p(1). \label{avar1} 
\end{equation}

\noindent {\bf Step 1.} By definition, $\hat\sigma^2(1) = n^{-2}k_n \sum_{i\in\N_n} \sum_{j\in\N_n} \hat{Z}_i \hat{Z}_j A_{ij}(1)$. We prove that
\begin{equation*}
  \hat\sigma^2(1) = \underbrace{\frac{k_n}{n^2} \sum_{i\in\N_n} \sum_{j\in\N_n} Z_i Z_j A_{ij}(1)}_{\check\sigma^2(1)} + o_p(1).
\end{equation*}

\noindent In the formula for $\hat\sigma^2(1)$, replace $\hat{Z}_i$ with $\hat{Z}_i \pm Z_i$ to obtain
\begin{multline}
  \abs{\hat\sigma^2(1) - \check\sigma^2(1)} \leq \frac{2k_n}{n^2} \sum_{i\in\N_n} \abs{Z_i} \sum_{j\in\N_n} A_{ij}(1) \max_k \abs{\hat{Z}_k - Z_k} \\ + \frac{k_n}{n^2} \sum_{i\in\N_n} \sum_{j\in\N_n} A_{ij}(1) \max_k (\hat{Z}_k - Z_k)^2. \label{2904yuws}
\end{multline}

\noindent By \autoref{lnbhd}, $\max_i \sum_j A_{ij}(1) \precsim n/k_n$. By \autoref{aboundY} and \autoref{lintersect}, $\max_i \abs{Z_i} \precsim 1$, and 
\begin{equation*}
  \max_i (\hat{Z}_i - Z_i)^2 = \max_i \left( \frac{T_{1i}(\mu_1 - \hat{\mu}_1)}{p_{1i}} - \frac{T_{0i}(\mu_0 - \hat{\mu}_0)}{p_{1i}} \right)^2 \precsim \max_t (\mu_t - \hat{\mu}_t)^2,
\end{equation*}

\noindent which is $\prec 1$ by the proof of \autoref{tQ}. Hence, $\eqref{2904yuws} \prec 1$.

\bigskip
\noindent {\bf Step 2.} We prove that
\begin{equation*}
  \check\sigma^2(1) = \frac{k_n}{n^2} \sum_{i\in\N_n} \sum_{j\in\N_n} (Z_i-\E[Z_i]) (Z_j-\E[Z_j]) A_{ij}(1) + \mathcal{B}_n(1) + o_p(1),
\end{equation*}

\noindent where
\begin{equation*}
  \mathcal{B}_n(1) = \frac{k_n}{n^2} \sum_{i\in\N_n} \sum_{j\in\N_n} \E[Z_i] \E[Z_j] A_{ij}(1).
\end{equation*}

\noindent In the formula of $\check\sigma^2(1)$, replace $Z_i$ with $Z_i \pm \E[Z_i]$ to obtain
\begin{multline*}
  \check\sigma^2(1) = \frac{k_n}{n^2} \sum_{i\in\N_n} \sum_{j\in\N_n} (Z_i-\E[Z_i]) (Z_j-\E[Z_j]) A_{ij}(1) \\ + \frac{2k_n}{n^2} \sum_{i\in\N_n} \sum_{j\in\N_n} (Z_i-\E[Z_i]) \E[Z_j] A_{ij}(1) + \frac{k_n}{n^2} \sum_{i\in\N_n} \sum_{j\in\N_n} \E[Z_i] \E[Z_j] A_{ij}(1).
\end{multline*}

\noindent We need to show that the second term on the right is $\prec 1$. For $V_i = \sum_{j\in\N_n} \E[Z_j] A_{ij}(1)$, 
\begin{multline}
  \E\bigg[ \bigg| \frac{k_n}{n^2} \sum_{i\in\N_n} \sum_{j\in\N_n} (Z_i-\E[Z_i]) \E[Z_j] A_{ij}(1) \bigg| \bigg] \leq \E\bigg[ \bigg( \frac{k_n}{n^2} \sum_{i\in\N_n} (Z_i-\E[Z_i]) V_i \bigg)^2 \bigg]^{1/2} \\
  = \bigg( \frac{k_n^2}{n^4} \sum_{i\in\N_n} \sum_{j\in\Lambda_i} \cov(Z_i,Z_j) V_i V_j + \frac{k_n^2}{n^4} \sum_{i\in\N_n} \sum_{j\not\in\Lambda_i} \cov(Z_i,Z_j) V_i V_j \bigg)^{1/2}. \label{vfneroi}
\end{multline}

\noindent By \autoref{aboundY} and \autoref{lintersect}, $\max_i\abs{\E[Z_i]} \precsim 1$. Since $\max_i \sum_{j\in\N_n} A_{ij}(1) \precsim n/k_n$ by \autoref{lnbhd}, $\max_i\abs{V_i} \precsim n/k_n$. Therefore,
\begin{equation*}
  \eqref{vfneroi} \precsim \bigg( \frac{1}{n^2} \sum_{i\in\N_n} \sum_{j\in\Lambda_i} \abs{\cov(Z_i,Z_j)} + \frac{1}{n^2} \sum_{i\in\N_n} \sum_{j\not\in\Lambda_i} \abs{\cov(Z_i,Z_j)} \bigg)^{1/2}.
\end{equation*}

\noindent The argument in \eqref{covzz}--\eqref{varP2.2} can be applied to show that this is $\prec 1$. The one distinction is that the above expression has $Z_i$ in place of $\tilde Z_i$. These only differ because $Y_i$ in the expression of $Z_i$ is centered by $\mu_t$, unlike the expression of $\tilde Z_i$, but the centering is immaterial since it cancels out when deriving the analogs of the covariance bounds \eqref{g902jghbwa} and \eqref{covzz0} using \eqref{e1290urt2g0}.

\bigskip
\noindent {\bf Step 3.} We prove that $\mathcal{B}_n(1) = \mathcal{B}_n + o_p(1)$. Noting that $C_{c(i)} \subseteq \Lambda_i$, decompose
\begin{equation*}
  \mathcal{B}_n(1) = \frac{k_n}{n^2} \sum_{i\in\N_n} \sum_{j\in C_{c(i)}} \E[Z_i] \E[Z_j] + \frac{k_n}{n^2} \sum_{i\in\N_n} \sum_{j\in \Lambda_i\backslash C_{c(i)}} \E[Z_i] \E[Z_j].
\end{equation*}

\noindent By \autoref{aclus}(a), $\Lambda_i \subseteq \N(i,2(r_n+U_n)) \subseteq \N(i,3U_n)$, and $\Lambda_i\backslash C_{c(i)} \subseteq \N(i,3U_n)\backslash \N(i,L_n)$. By \autoref{aclus}(b), there exists a non-negative sequence $\alpha_n \precsim 1$ such that $3U_n - L_n = \alpha_n (n/(k_n\xi_n))^{1/d}$. Then by \autoref{aspace}(b), following \eqref{precboundary},
\begin{equation*}
  \max_i \abs{\N(i,3U_n)\backslash \N(i,L_n)} \precsim \xi_n (n/(k_n\xi_n))^{(d-1)/d}.
\end{equation*}

By \autoref{aboundY} and \autoref{lintersect}, $\max_i \abs{\E[Z_i]} \precsim 1$, so 
\begin{equation*}
  \frac{k_n}{n^2} \sum_{i\in\N_n} \sum_{j\in\Lambda_i\backslash C_{c(i)}} \E[Z_i] \E[Z_j] \precsim \frac{k_n}{n} \max_i \abs{\N(i,3U_n)\backslash \N(i,L_n)} \precsim \frac{k_n\xi_n}{n} \left( \frac{n}{k_n\xi_n} \right)^{(d-1)/d},
\end{equation*}

\noindent which is $\prec 1$ since $k_n \prec n/\xi_n$ by assumption.

\bigskip
\noindent {\bf Step 4.} Abbreviate $\F_i \equiv \F_i(r_n)$. At the top of the proof, we noted that
\begin{equation*}
  \sigma_n^2 = \frac{k_n}{n^2} \sum_{i\in\N_n} \sum_{j\in\N_n} \cov(\E[Z_i \mid \F_i], \E[Z_j \mid \F_j]) A_{ij}(1) + o_p(1).
\end{equation*}

\noindent By steps 1--3, 
\begin{equation*}
  \hat\sigma^2(1) = \frac{k_n}{n^2} \sum_{i\in\N_n} \sum_{j\in\N_n} (Z_i-\E[Z_i]) (Z_j-\E[Z_j]) A_{ij}(1) + \mathcal{B}_n + o_p(1)
\end{equation*}

\noindent It therefore remains to show that the following is $\prec 1$:
\begin{align*}
  &\begin{aligned} \frac{k_n}{n^2} \sum_{i\in\N_n} \sum_{j\in\N_n} (Z_i-\E[Z_i]) &(Z_j-\E[Z_j]) A_{ij}(1) \\ 
  &- \frac{k_n}{n^2} \sum_{i\in\N_n} \sum_{j\in\N_n} \cov(\E[Z_i \mid \F_i], \E[Z_j \mid \F_j]) A_{ij}(1) \end{aligned} \\ 
  &\begin{aligned}=\frac{k_n}{n^2} \sum_{i\in\N_n} \sum_{j\in\N_n} (Z_iZ_j - \E[\E[Z_i \mid \F_i] &\E[Z_j \mid\, \F_j]]) A_{ij}(1) \\ 
  &- 2\frac{k_n}{n^2} \sum_{i\in\N_n} (Z_i-\E[Z_i]) \sum_{j\in\N_n} \E[Z_j] A_{ij}(1). \end{aligned} \\
  &\equiv [P1] + [P2].
\end{align*}

\noindent As previously argued, $\max_i \sum_{j\in\N_n} \E[Z_j] A_{ij}(1) \precsim n/k_n$, and $\abs{n^{-1} \sum_{i\in\N_n} (Z_i-\E[Z_i])} \prec 1$ by the proof of \autoref{tQ}, so $[P2] \prec 1$, while
\begin{multline*}
  [P1] = \frac{k_n}{n^2} \sum_{i\in\N_n} \sum_{j\in\N_n} (Z_i-\E[Z_i \mid \F_i]) (Z_j-\E[Z_j \mid \F_j]) A_{ij}(1) \\
  + 2\frac{k_n}{n^2} \sum_{i\in\N_n} \sum_{j\in\N_n} \E[Z_i \mid \F_i] (Z_j-\E[Z_j \mid \F_j]) A_{ij}(1) \\
  + \frac{k_n}{n^2} \sum_{i\in\N_n} \sum_{j\in\N_n}  (\E[Z_i \mid \F_i] \E[Z_j \mid \F_j] - \E[\E[Z_i \mid \F_i] \E[Z_j \mid \F_j]]) A_{ij}(1) \\
  \equiv [P1.1] + [P1.2] + [P1.3].
\end{multline*}

\noindent Using \eqref{r_nbound}, \eqref{bhn0ejr5}, \autoref{lnbhd}, and the assumption that $k_n \prec n/\xi_n$,
\begin{equation*}
  [P1.1] \precsim \frac{k_n}{n^2} \cdot n \cdot \frac{n}{k_n} \cdot r_n^{-\gamma} \prec 1 \quad\text{and}\quad [P1.2] \precsim \frac{k_n}{n^2} \cdot n \cdot \frac{n}{k_n} \cdot r_n^{-\gamma} \prec 1.
\end{equation*}

Observe that $[P1.3]$ has expectation zero and variance equal to
\begin{equation}
  \frac{k_n^2}{n^4} \sum_{i\in\N_n} \sum_{j\in\N_n} \sum_{k\in\N_n} \sum_{\ell\in\N_n} \cov(\E[Z_i \mid \F_i] \E[Z_j \mid \F_j], \E[Z_k \mid \F_k] \E[Z_\ell \mid \F_\ell]) A_{ij}(1) A_{k\ell}(1). \label{fjio32}
\end{equation}

\noindent The covariance term is zero if $k,\ell \not\in \Lambda_i \cup \Lambda_j$, so by \autoref{aboundY} and \autoref{lintersect}, 
\begin{equation}
  \eqref{fjio32} \precsim \frac{k_n^2}{n^4} \sum_{i\in\N_n} \sum_{j\in\Lambda_i} \abs{\Xi_{ij}} \label{bj3h34}
\end{equation}

\noindent for $\Xi_{ij} = \{(k,\ell)\colon \ell\in\Lambda_k \text{ and } \{k,\ell\} \cap (\Lambda_i \cup \Lambda_j) \neq \emptyset\}$. Since $\max_i \abs{\Lambda_i} \precsim n/k_n$ and $\ell\in\Lambda_k$ implies $k \in \Lambda_\ell$, we have $\max_{i,j\in\mathcal{N}_n} \abs{\Xi_{ij}} \precsim (n/k_n)^2$. Therefore,
\begin{equation*}
  \eqref{bj3h34} \precsim \frac{k_n^2}{n^4} \cdot n \cdot \frac{n}{k_n} \cdot \frac{n^2}{k_n^2} \precsim k_n^{-1} \prec 1.
\end{equation*}

\noindent {\bf Step 5.} Recall the definition of $\tilde\sigma_n^2$ from \eqref{tildesigma}. We next prove that
\begin{equation*}
  \tilde\sigma_n^2 = \frac{k_n}{n^2} \sum_{i\in\N_n} \sum_{j\in\N_n} \cov(\E[Z_i \mid \F_i], \E[Z_j \mid \F_j]) A_{ij}(2) + o_p(1).
\end{equation*}

\noindent Then the claimed result for $\hat\sigma^2(2)$ follows from steps 1, 2, and 4 by replacing, for all $i,j$, every occurrence of ``$A_{ij}(1)$'' and ``$\Lambda_i$'' with ``$A_{ij}(2)$'' and ``$C_{c(i)}$'', respectively. Write
\begin{multline*}
  \tilde\sigma_n^2 = \frac{k_n}{n^2} \sum_{i\in\N_n} \sum_{j\in C_{c(i)}} \cov(\E[Z_i \mid \F_i], \E[Z_j \mid \F_j]) \\ + \frac{k_n}{n^2} \sum_{i\in\N_n} \sum_{j\in \Lambda_i\backslash C_{c(i)}} \cov(\E[Z_i \mid \F_i], \E[Z_j \mid \F_j]).
\end{multline*}

\noindent By \autoref{aboundY} and \autoref{lintersect}, $\max_{i,j} \abs{\cov(\E[Z_i \mid \F_i], \E[Z_j \mid \F_j])} \precsim 1$, so using the argument in step 3,
\begin{multline*}
  \bigg| \frac{k_n}{n^2} \sum_{i\in\N_n} \sum_{j\in\Lambda_i\backslash C_{c(i)}} \cov(\E[Z_i \mid \F_i], \E[Z_j \mid \F_j]) \bigg| \\ \leq \frac{k_n}{n} \max_i \abs{\N(i,3U_n)\backslash\N(i,L_n)} \precsim \frac{k_n\xi_n}{n} \left( \frac{n}{k_n\xi_n} \right)^{(d-1)/d} \prec 1.
\end{multline*} \qed

%---------------------------------------
\subsection{\autoref{lsep}}\label{slep}
%---------------------------------------

We draw on the clever argument used in the proof of Theorem 3 of \cite{cao2024inference}. To obtain a contradiction, suppose there is a positive sequence $\ell_n \prec 1$ such that, for infinitely many $n$, there are two clusters $C_1$ and $C_2$ with medoids $i_1,i_2$ such that $\rho(i_1,i_2) < \ell_n (n/(k_n\xi_n))^{1/d}$, and no other pair of medoids is closer in distance. There are two cases to consider along the subsequence of such $n$'s, and all asymptotic statements that follow are with respect to this subsequence.

{\bf Case 1.} $\min\{\abs{C_1}, \abs{C_2}\} \precsim n/k_n$. Then the argument largely proceeds as in the proof of Theorem 3 of \cite{cao2024inference} but with additional adjustments to account for the divergence of $k_n,\xi_n$. Since clusters partition $\N_n$, for every $n$ there must exist some cluster $C_3$ of size at least $n/k_n$. Let $i_3$ denote its medoid and $R_3$ be the largest distance between $i_3$ and an element in $C_3$. By \autoref{aspace}(a), 
\begin{equation}
  \frac{n}{k_n} \leq \abs{C_3} \leq \abs{\N(i_3,R_3)} \leq C \xi_n R_3^d \quad\Longrightarrow\quad R_3 \succsim (n/(k_n\xi_n))^{1/d}. \label{bar_r3}
\end{equation}

Let $i_3' \in C_3$ be such that $\rho(i_3',j) \geq 0.5 R_3$ for any medoid $j$. For instance, the unit $i_3^* \in C_3$ for which $\rho(i_3^*,i_3)=R_3$ is one such candidate. For any other medoid $j\neq i_3$, $\rho(i_3',j) \geq 0.5R_3$ because otherwise step 1 of \autoref{amedoids} would have assigned $i_3'$ to a closer medoid.

Consider a hypothetical update in step 2 of \autoref{amedoids} that replaces $i_2$ with $i_3'$. Then all units in $\N(i_3',R_3/8)$ are optimally reassigned to the cluster with medoid $i_3'$ because they are all by construction at least distance $3/8\cdot R_3$ away from any other medoid. This reassignment reduces the total cost by at least $\abs{\N(i_3',R_3/8)} R_3/4 \succsim n/k_n \cdot (n/(k_n\xi_n))^{1/d}$ by \autoref{aspace}(a) and \eqref{bar_r3}. On the other hand, in the worst case, all other units in $C_2$ are reassigned to medoid $i_1$, which increases the total cost by at most $\abs{C_2} \ell_n (n/(k_n\xi_n))^{1/d} \prec n/k_n \cdot (n/(k_n\xi_n))^{1/d}$. Hence, the update is overall cost-reducing for $n$ sufficiently large, which contradicts for such $n$ the supposition that the clusters are the output of \autoref{amedoids}, in particular violating step 2.

{\bf Case 2.} $\min\{\abs{C_1}, \abs{C_2}\} \succ n/k_n$. Without loss of generality, suppose $\abs{C_2} \succ n/k_n$. Then the argument proceeds similarly to case 1 but using $C_2$ in place of $C_3$. In particular, let $R_2$ be the largest distance between $i_2$ and an element of $C_2$. By an argument similar to \eqref{bar_r3}, $R_2 \succ (n/(k_n\xi_n))^{1/d}$. 

Let $i_2' \in C_2$ be such that $\rho(i_2',j) \geq 0.5 R_2$ for any medoid $j$. Consider a hypothetical update in step 2 of \autoref{amedoids} that replaces $i_2$ with $i_2'$. Then all units in $\N(i_2',R_2/8)$ are optimally reassigned to the cluster with medoid $i_2'$. This reassignment reduces the total cost by at least $\abs{\N(i_2',R_2/8)} R_2/4 \succ \abs{\N(i_2',R_2/8)} (n/(k_n\xi_n))^{1/d}$. On the other hand, in the worst case, all other units in $C_2$ are reassigned to medoid $i_1$, which increases the total cost by at most $\abs{C_2} \ell_n (n/(k_n\xi_n))^{1/d}$. Since $\abs{C_2} \leq \abs{\N(i_2,R_2)} \sim \abs{\N(i_2',R_2/8)}$ by \autoref{aspace}(a), the update is overall cost-reducing for $n$ sufficiently large, which contradicts the supposition that the clusters are the output of \autoref{amedoids}.  \qed

%---------------------------------------
\subsection{\autoref{lradius}}\label{slradius}
%---------------------------------------

{\bf Step 1.} We first prove that, for any sequence $\{C_n\}_{n\in\mathbb{N}}$ with $C_n \in \C_n$ for each $n$, we have $\abs{C_n} \sim n/k_n$. Let $\ell_n$ be given as in \autoref{lsep} and $i_n$ the medoid of cluster $C_n$. By \autoref{aspace}(a), $\abs{\N(i_n,0.5 \ell_n (n/(k_n\xi_n))^{1/d})} \succsim n/k_n$. All units in this neighborhood must be elements of $C_n$ by \autoref{lsep} and step 1 of \autoref{amedoids}, so $\abs{C_n} \succsim n/k_n$. Since this must be true for all sequences of clusters, it cannot be the case that $\abs{C_n} \succ n/k_n$, given that the total number of units is $n$. Hence $\abs{C_n} \sim n/k_n$. 

\bigskip
\noindent {\bf Step 2.} We prove the direction $R_n^* \precsim (n/(k_n\xi_n))^{1/d}$. Suppose to the contrary that $R_n^* \succ (n/(k_n\xi_n))^{1/d}$. Let $i_n' \in C_n$ be such that $\rho(i_n',j_n) \geq 0.5 R_n^*$ for any medoid $j_n$. For instance, the unit $i_n^* \in C_n$ for which $\rho(i_n^*,i_n)=R_n^*$ is one such candidate. For any other medoid $j_n\neq i_n$, $\rho(i_n',j_n) \geq 0.5R_n^*$ because otherwise it would be cost-reducing for \autoref{amedoids} to assign $i_n'$ to a closer medoid.

Consider a hypothetical update in step 2 \autoref{amedoids} that replaces $i_n$ with $i_n'$. Then all units in $\N(i_n',R_n^*/8)$ are optimally reassigned to the cluster with medoid $i_n'$. This reassignment reduces the total cost by at least $\abs{\N(i_n',R_n^*/8)} R_n^*/4 \succ n/k_n (n/(k_n\xi_n))^{1/d}$ by \autoref{aspace}(a). On the other hand, in the worst case, all other units in $C_n$ are reassigned to the nearest existing medoid. 

Observe that the distance between $i_n$ and the nearest other medoid must be $\precsim (n/(k_n\xi_n))^{1/d}$. If instead that distance were $\delta_n \succ (n/(k_n\xi_n))^{1/d}$, then $C_n$ would contain $\N(i_n, 0.5\delta_n)$ by step 1 of \autoref{amedoids}, in which case it would have size $\succ n/k_n$ by \autoref{aspace}(a), which contradicts step 1.

Therefore, the worst-case increase in total cost from the medoid replacement is $\precsim \abs{C_n} (n/(k_n\xi_n))^{1/d} \precsim n/k_n \cdot (n/(k_n\xi_n))^{1/d}$ by step 1. The update is overall cost-reducing for $n$ sufficiently large, which contradicts the supposition that the clusters are the output of \autoref{amedoids}.

\bigskip
\noindent {\bf Step 3.} We prove the other direction. By \autoref{lsep} and step 1 of \autoref{amedoids}, $\N(i_n, \ell_n (n/(k_n\xi_n))^{1/d}/2) \subseteq C_n$. Hence, $R_n^* \geq \ell_n (n/(k_n\xi_n))^{1/d}/2$, so $R_n^* \succsim (n/(k_n\xi_n))^{1/d}$. \qed

%---------------------------------------
\subsection{\autoref{lintersect}}\label{slintersect}
%---------------------------------------

Recall from \autoref{aclus} the definition of cluster centroids, $U_n$, and $L_n$. For any $i\in\N_n$, all clusters intersecting $\N(i,r_n)$ must be subsets of $\N(i, r_n + 2U_n)$ by the assumption. Together with \eqref{kappa}, we have that there exists $\alpha_n \sim 1$ such that $\N(i, r_n + 2U_n) \subseteq \N(i,\alpha_nL_n)$ for any $n$ and $i\in\N_n$. Invoking \autoref{aclus} once more, it follows that $\phi_i$ is at most the number of $L_n$-balls that fill $\N(i, \alpha_nL_n)$ without intersecting. 

We seek to bound the $2L_n$-packing number \citep[][Definition 5.4]{wainwright2019high} of $\N(i,\alpha_nL_n)$. By Lemma 5.5 of \cite{wainwright2019high}, this is at most the $L_n$-covering number of the ball \citep[][Definition 5.1]{wainwright2019high}, and we denote this number by $N_n(i)$. We bound this following the proof of Lemma 5.7(b) in \cite{wainwright2019high}. 

Construct a maximal $L_n/2$-packing of $\N(i,\alpha_n L_n)$ with cardinality $M_i$ and centroids $\{\theta_m\}_{m=1}^{M_i}$. This is also an $L_n$-covering of $\N(i,\alpha_n L_n)$, so $N_n(i) \leq M_i$. The balls of the packing $\{\N(\theta_m, L_n/2)\}_{m=1}^{M_i}$ are disjoint and contained in $\N(i,\alpha_n L_n + L_n/2)$, so
\begin{equation*}
  \sum_{m=1}^{M_i} \abs{\N(\theta_m,L_n/2)} \leq \abs{\N(i,L_n(\alpha_n+0.5))}.
\end{equation*}

\noindent By \autoref{aspace}(a), there exists $C>0$ independent of $i$ and $n$ such that
\begin{align*}
  &\sum_{m=1}^{M_i} \abs{\N(\theta_m,L_n/2)} \geq \frac{M_i}{C} \xi_n \left( \frac{L_n}{2} \right)^d \quad\text{and} \\
  &\abs{\N(i,L_n(\alpha_n+0.5))} \leq C \xi_n \left( L_n(\alpha_n+0.5) \right)^d.
\end{align*}

\noindent Hence, $M_i \leq C^2 (1+2\alpha_n)^d \precsim 1$, so $\max_i N_n(i) \precsim 1$, which proves the first claim of the lemma.

The second claim follows from the expression in \autoref{roverlap}, the first claim, and \autoref{aoverlap}. \qed

%---------------------------------------
\subsection{\autoref{lnbhd}}\label{slnbhd}
%---------------------------------------

Let $\bar{R} = \max_j R_j$, the latter defined prior to \eqref{kappa}. Observe that $\Lambda_i \subseteq \N(i,r_n + 2\bar{R} + r_n)$. By \autoref{aclus} and \eqref{kappa}, $r_n \precsim \bar{R} \precsim (n/(k_n\xi_n))^{1/d}$, so by \autoref{aspace}(a), $\max_i \abs{\Lambda_i} \precsim n/k_n$. \qed

%---------------------------------------
\subsection{\autoref{lkappa}}\label{slkappa}
%---------------------------------------

Since $\hat\kappa_t$ has mean one, it remains to compute the variance. By \autoref{aCRT},
\begin{equation*}
  \var(\hat\kappa_t) = \frac{1}{n^2} \sum_{j=1}^{k_n} \var\left( \sum_{i\in C_j} \frac{T_{ti}}{p_{ti}} \right)
\end{equation*}

\noindent By \autoref{lintersect}, $\min_i p_{ti} \succsim 1$, and by Assumptions \ref{aspace}(a) and \ref{aclus}, $\max_j \abs{C_j} \precsim n/k_n$, so the right-hand side is at most of order
\begin{equation*}
  \frac{1}{n^2} k_n \max_j \abs{C_j}^2 \precsim \frac{1}{k_n}.
\end{equation*} \qed

%---------------------------------------
\subsection{\autoref{lpt}}\label{slpt}
%---------------------------------------

By \autoref{aoverlap}, $p_t^+ \in (0,1)$. Since
\begin{equation*}
  \frac{p_t^+}{\hat{p}_t^+} - 1 = \frac{p_t^+ - \hat{p}_t^+}{p_t^+ + \hat{p}_t^+ - p_t^+},
\end{equation*}

\noindent it is enough to show that $\hat{p}_t^+ - p_t^+ \precsim k_n^{-1/2}$. Clearly $\E[\hat{p}_t^+] = p_t^+$. By \autoref{aCRT}, 
\begin{equation*}
  \var\big( \hat{p}_t^+ \big) = \frac{1}{n^2} \sum_{j=1}^{k_n} \var\bigg( \sum_{i \in C_j} T_{ti}^+ \bigg) \leq \frac{1}{n^2} \sum_{j=1}^{k_n} 2\abs{C_j}^2.
\end{equation*}

\noindent By Assumptions \ref{aspace}(a) and \ref{aclus}, $\max_j \abs{C_j} \precsim n/k_n$, so the right-hand side of the above display is $\precsim n^{-2} k_n (n/k_n)^2 = k_n^{-1}$. \qed

%----------------------------------------------------------------------

\FloatBarrier
\phantomsection
\addcontentsline{toc}{section}{References}
\bibliography{CRT}{} 
\bibliographystyle{aer}

%----------------------------------------------------------------------

\end{document}